\documentclass[10pt,reqno]{amsart}

\usepackage{amsmath,amssymb,amsfonts,dsfont}
\usepackage{cite,graphicx,xcolor,hyperref}

\theoremstyle{plain}
  \newtheorem{Theorem}{Theorem}
  \newtheorem{Proposition}{Proposition}
  \newtheorem{Algorithm}{Algorithm}

\theoremstyle{remark}
  \newtheorem{Remark}{Remark}
  \newtheorem{Example}{Example}

\author{T.\,A. Averina\,${}^1$, K.\,A. Rybakov\,${}^2$}

\title[Analysis and conditional optimization of projection estimates for distribution\dots]{Analysis and conditional optimization of projection estimates for distribution of random variable using~Legendre~polynomials}

\begin{document}

\maketitle

\begin{center}

\vskip -3.5ex

${}^1$\,Institute of Computational Mathematics and Mathematical Geophysics,\\
Siberian Branch of the Russian Academy of Sciences;\\
630090, Novosibirsk, Akad. Lavrentyev Ave, 6; ata@osmf.sscc.ru\\

\vskip 1ex

${}^2$\,Moscow Aviation Institute (National Research University);\\
125993, Moscow, Volokolamskoe Hwy, 4; rkoffice@mail.ru
\end{center}

\vskip 2.5ex

\textbf{Abstract.} Algorithms for jointly obtaining projection estimates of the density and distribution function of a random variable using Legendre polynomials are proposed. For these algorithms, a problem of the conditional optimization is solved. Such optimization allows one to increase the approximation accuracy with minimum computational costs. The proposed algorithms are tested on examples with different degrees of smoothness of the density. A projection estimate of the density is compared to a histogram that is often used in applications to estimate distributions.

\vskip 0.5ex


\vskip 0.5ex


\makeatletter{\renewcommand*{\@makefnmark}{}
\footnotetext{Citation: Averina, T.; Rybakov, K. Analysis and conditional optimization of projection estimates for distribution of random variable using Legendre polynomials. {\em Algorithms} {\bf 2025}, {\em 18(8)}, 466. \url{https://doi.org/10.3390/a18080466}}
\footnotetext{This research study was carried out within the framework of the state assignment to the Institute of Computational Mathematics and Mathematical Geophysics, Siberian Branch, Russian Academy of Sciences (project FWNM-2025-0002).}}

\section{Introduction}\label{secIntro}

Statistical projection estimates in the Monte Carlo method were first proposed by N.N.~Chentsov~\cite{Che_DAN62}. He developed a general technique for optimizing such estimates; this technique, however, requires clarification in specific problems~\cite{Mih_NAA19}. In the paper~\cite{Ave_Math24}, projection estimates based on the Legendre polynomials for marginal probability densities of solutions to stochastic differential equations were proposed. The mean square error of the estimates was studied, and a comparison of both the obtained projection estimates and the histogram was carried out with examples. Analysis of the results showed that, on the same sample size, the projection estimate approximates the density more accurately. In addition, such an estimate is specified analytically, and it is a smooth function. So, it is preferable in, e.g., filtering problems and the control of nonlinear systems~\cite{AveRyb_NAA17, Rud_JCSSI23, AveRyb_RJNAM20, PanKar_Math23}.

This paper presents two statistical algorithms, based on algorithms using Legendre polynomials from~\cite{Ave_Math24}, for jointly obtaining projection estimates of the density and distribution function of a random variable.

When solving different problems by statistical estimation algorithms, it is important to make an optimal (consistent) choice of their parameters for finding the mathematical expectation of a certain functional that depends on a random variable. Therefore, this paper analyzes the mean square error of the projection estimate from this point of view. Such parameters are the projection expansion length and sample size. We solve a conditional optimization problem considered by G.A.~Mikhailov~\cite{Mih_00}. The objective of this problem is to minimize the algorithm complexity while achieving the required level of approximation accuracy. We study how to minimize the mean square error of the projection estimates of the density and distribution function by the equalization of its deterministic and stochastic components. The accuracy of the projection estimates also depends on the degree of smoothness of the density; therefore, in this paper, we consider the dependence of error not only on the projection expansion length but also on the degree of smoothness of the approximated function.

The obtained theoretical results are confirmed with examples using a two-parameter family of densities that allows one both to choose the degree of smoothness and to perform a simple calculation of expansion coefficients with respect to Legendre polynomials.

The rest of this paper has the following structure: Section~\ref{secLeg} contains the necessary information on the Fourier--Legendre series and the definition of projection estimates of the density and distribution function; in this section, relations for expansion coefficients of the density and distribution function with respect to Legendre polynomials are obtained. Section~\ref{secMonteCarlo} presents algorithms for jointly obtaining randomized projection estimates of the density and distribution function. The analysis and conditional optimization of randomized projection estimates are carried out in Section~\ref{secOptim}. Section~\ref{secTest} proposes a two-parameter family of densities with different degrees of smoothness and related distribution functions; it presents an algorithm for modeling the corresponding random variables, gives expansion coefficients of the considered densities and distribution functions, and studies the convergence rate of their expansions. Numerical experiments and their analysis are discussed in Section~\ref{secNumerics}. In Section~\ref{secHistogram}, a projection estimate of the density is compared with a histogram. Brief conclusions on the paper are given in Section~\ref{secConcl}.

\section{Using Legendre Polynomials for Randomized Projection Estimates}\label{secLeg}

\subsection{Fourier--Legendre Series}

The standardized Legendre polynomials $\{P_i\}_{i = 0}^\infty$ are defined on the set $\Omega = [-1,1]$ as follows~\cite{Sue_79}:
\[
  P_i(x) = \frac{1}{2^i i!} \, \frac{d^i (x^2-1)^i}{dx^i}, \ \ \ i = 0,1,2,\dots
\]

They form a complete orthogonal system of functions in the space $L_2(\Omega)$~\cite{Tri_10}, where
\[
  L_2(\Omega) = \biggl\{ u \colon \int_\Omega |u(x)|^2 dx < \infty \biggr\},
\]
and satisfy well-known recurrence relations:
\begin{gather}
  (i+1) P_{i+1}(x) = (2i+1) x P_i(x) - i P_{i-1}(x), \ \ \ i = 1,2,\dots, \label{eqNNLegRec1} \\
  (2i+1) P_i(x) = P'_{i+1}(x) - P'_{i-1}(x), \ \ \ i = 1,2,\dots, \label{eqNNLegRec2}
\end{gather}
where $P_0(x) = 1$ and $P_1(x) = x$. These relations can be formally applied for $i = 0$ when $P_{-1}(x) = 0$.

The explicit formula for the standardized Legendre polynomials is
\begin{equation}\label{eqLegExp}
  P_i(x) = \frac{1}{2^i} \sum\limits_{k = 0}^{\lfloor i/2 \rfloor} \frac{(-1)^k(2i-2k)!}{k!(i-k)!(i-2k)!} \, x^{i-2k}, \ \ \ i = 0,1,2,\dots,
\end{equation}
where $\lfloor \,\cdot\, \rfloor$ is the floor function. The explicit formulae can be specified separately for even indices ($i = 2j$), as
\begin{equation}\label{eqLegExpA}
  P_{2j}(x) = \frac{1}{2^{2j}} \sum\limits_{k = 0}^j \frac{(-1)^k(4j-2k)!}{k!(2j-k)!(2j-2k)!} \, x^{2j-2k},
\end{equation}
and for odd indices ($i = 2j+1$), as
\begin{equation}\label{eqLegExpB}
  P_{2j+1}(x) = \frac{1}{2^{2j+1}} \sum\limits_{k = 0}^j \frac{(-1)^k(4j-2k+2)!}{k!(2j-k+1)!(2j-2k+1)!} \, x^{2j-2k+1}.
\end{equation}

Using the normalization, we obtain the Legendre polynomials $\{\hat P_i\}_{i = 0}^\infty$ that form a complete orthonormal system of functions in $L_2(\Omega)$:
\begin{equation}\label{eqDefLeg}
  \hat P_i(x) = \sqrt{\frac{2i+1}{2}} \, P_i(x), \ \ \ i = 0,1,2,\dots
\end{equation}

The expansion of an arbitrary function $u \in L_2(\Omega)$ into the Fourier--Legendre series has the form
\begin{equation}\label{eqFourier}
  u(x) = \sum\limits_{i=0}^\infty U_i \hat P_i(x),
\end{equation}
where coefficients $U_i$ are defined as inner products, i.e.,
\[
  U_i = (u,\hat P_i)_{L_2(\Omega)} = \int_\Omega u(x) \hat P_i(x) dx, \ \ \ i = 0,1,2,\dots
\]

This expansion is a basis for the approximation of $u$:
\begin{equation}\label{eqFourierTrunc}
  u(x) \approx u^{\langle n \rangle}(x) = \sum\limits_{i=0}^n U_i \hat P_i(x),
\end{equation}
and the approximation accuracy is estimated depending on the chosen metric (or norm). For example, in $L_2(\Omega)$, we have~\cite{CanHusQuaZan_06}
\begin{equation}\label{eqEstimateL2A}
  \| u - u^{\langle n \rangle} \|_{L_2(\Omega)} \leqslant \frac{C}{n^s},
\end{equation}
where $C > 0$ does not depend on $n$, under the condition that the generalized derivatives of $u$ up to order $s$ belong to $L_2(\Omega)$. There is also the sharper estimate~\cite{Schwab_98, HouSchwabSuli_JNA02}
\begin{equation}\label{eqEstimateL2B}
  \| u - u^{\langle n \rangle} \|_{L_2(\Omega)} \leqslant \tilde C \, \biggl( \frac{(n-s+1)!}{(n+s+1)!} \biggr)^{1/2},
\end{equation}
for which it is assumed that $\tilde C > 0 $ does not depend on $n$, and the generalized derivatives of $u$ up to order $s$ belong to $L_2(\Omega)$ with weight $\rho(x) = (1-x^2)^s$, $s \leqslant n+1$.

We also present the estimate corresponding to $C(\Omega)$, the space of all continuous functions on $\Omega$. If a function $u$ is continuously differentiable $r$ times on $\Omega$, i.e., $u \in C^r(\Omega)$, and its derivative $u^{(r)}$ satisfies the Lipschitz condition with parameter $\alpha \in (0,1]$, then
\[
  | u(x) - u^{\langle n \rangle}(x) | \leqslant \frac{\hat C}{n^{r+\alpha-1/2}} \ \ \ \forall \, x \in \Omega,
\]
where $\hat C > 0$ does not depend on $n$, $r + \alpha > 1/2$~\cite{Sue_79}.

\begin{Remark}\label{remSreal}
In inequality~\eqref{eqEstimateL2A}, the parameter $s$ can be a real non-negative number~\cite{CanQua_MC82}. In this case, we should use the condition $u \in W_2^s(\Omega)$ for $s \geqslant 0$ instead of the condition that the generalized derivatives of $u$ up to order $s$ belong to $L_2(\Omega)$. If $s$ is a natural number, then $W_2^s(\Omega)$, the Sobolev space, is defined in the usual way~\cite{Tri_10}:
\[
  W_2^s(\Omega) = \bigl\{ u \colon u^{(k)} \in L_2(\Omega), ~ k = 0,1,\dots,s \bigr\}.
\]

If $s$ is not a natural number, then it is possible to define $W_2^s(\Omega)$, the Sobolev--Slobodetskij space, by the complex interpolation between $W_2^{\lfloor s \rfloor}(\Omega)$ and $W_2^{\lceil s \rceil}(\Omega)$~\cite{Tri_10, Agr_13}, where $\lceil s \rceil = \lfloor s \rfloor + 1$ and $W_2^0(\Omega) = L_2(\Omega)$. However, we use another definition for $W_2^s(\Omega)$ from~\cite{Slo_DAN58, Agr_13, EdmEva_23}:
\begin{equation}\label{eqDefWs}
  W_2^s(\Omega) = \biggl\{ u \in L_2(\Omega) \colon \int_{\Omega^2} \frac{|u(x)-u(y)|^2}{|x-y|^{1+2s}} \, dx dy < \infty \biggr\}, \ \ \ s = \sigma \in (0,1),
\end{equation}
and
\[
  W_2^s(\Omega) = \bigl\{ u \in W_2^{\lfloor s \rfloor}(\Omega), \ \ \ u^{(\lfloor s \rfloor)} \in W_2^\sigma(\Omega), \ \ \ \sigma = s-{\lfloor s \rfloor} \bigr\}, \ \ \ s \geqslant 1.
\]

The parameter $s$ in inequalities~\eqref{eqEstimateL2A} and~\eqref{eqEstimateL2B} may be related to the condition that derivatives of fractional order $s$ for $u$ belong to $L_2(\Omega)$. For non-integer $s$, we should replace the ratio $(n-s+1)!/(n+s+1)!$ with the ratio $\Gamma(n-s+2)/\Gamma(n+s+2)$ in inequality~\eqref{eqEstimateL2B}, where $\Gamma$ is the gamma function.
\end{Remark}

\subsection{Projection Estimates of Density and Distribution Functions}

Let $\xi$ be the random variable with range $\Omega = [-1,1]$, density $g$, and distribution function $f$. The function $g$ equals to zero outside $\Omega$, and the function $f$ is expressed through $g$~\cite{Cra_99}:
\begin{equation}\label{eqDistrib}
  f(x) = \int_{-\infty}^x g(\theta) d\theta.
\end{equation}

Further, we consider projection estimates of $g,f \in L_2(\Omega)$ using their expansions into the Fourier--Legendre series (the restriction of the density and distribution function to $\Omega$ is implied). Note that the condition $f \in L_2(\Omega)$ follows from the condition $g \in L_2(\Omega)$, which is assumed to be satisfied.

As a function $u$ in the above formulae, we can use the density $g$ and distribution function $f$ (below, we denote the corresponding expansion coefficients by $G_i$ and $F_i$). In this case, the sequence $f^{\langle n \rangle}$ converges to $f$ faster than the sequence $g^{\langle n \rangle}$ converges to $g$, where
\begin{equation}\label{eqGFApprox}
  g^{\langle n \rangle}(x) = \sum\limits_{i=0}^n G_i \hat P_i(x), \ \ \ f^{\langle n \rangle}(x) = \sum\limits_{i=0}^n F_i \hat P_i(x),
\end{equation}
since the degree of smoothness of $f$ is greater by one than the degree of smoothness of $g$ due to their relationship~\eqref{eqDistrib}.

For example, let $g$ be continuous on $\Omega$, differentiable in a generalized sense only, i.e., $s = 1$, $r = 0$, and satisfy the Lipschitz condition with parameter $\alpha = 1$. Then
\begin{gather*}
  \| g - g^{\langle n \rangle} \|_{L_2(\Omega)} \leqslant \frac{C^g}{n} \ \ \ \text{or} \ \ \
  \| g - g^{\langle n \rangle} \|_{L_2(\Omega)} \leqslant \frac{\tilde C^g}{\sqrt{(n+1)(n+2)}}, \\
  | g(x) - g^{\langle n \rangle}(x) | \leqslant \frac{\hat C^g}{n^{1/2}},
\end{gather*}
and for $f$ the following estimates hold ($s = 2$, $r = 1$, $\alpha = 1$):
\begin{gather*}
  \| f - f^{\langle n \rangle} \|_{L_2(\Omega)} \leqslant \frac{C^f}{n^2} \ \ \ \text{or} \ \ \
  \| f - f^{\langle n \rangle} \|_{L_2(\Omega)} \leqslant \frac{\tilde C^f}{\sqrt{n(n+1)(n+2)(n+3)}}, \\
  | f(x) - f^{\langle n \rangle}(x) | \leqslant \frac{\hat C^f}{n^{3/2}},
\end{gather*}
where $C^g,\tilde C^g,\hat C^g,C^f,\tilde C^f,\hat C^f > 0$ are some constants and $x \in \Omega$.

The projection estimates of the density and distribution function are given by expressions corresponding to their approximations~\eqref{eqGFApprox}:
\begin{equation}\label{eqGFProj}
  \bar{g}^{\langle n \rangle}(x) = \sum\limits_{i=0}^n \bar{G}_i \hat P_i(x), \ \ \ \bar{f}^{\langle n \rangle}(x) = \sum\limits_{i=0}^n \bar{F}_i \hat P_i(x),
\end{equation}
where $\bar{G}_i$ and $\bar{F}_i$ are estimates of expansion coefficients $G_i$ and $F_i$ based on observations of the random variable $\xi$~\cite{Che_72, Ibr_ZNSP13}, $i = 0,1,\dots,n$.

\subsection{Expansion Coefficients of Density and Distribution Function}

The relation between expansion coefficients $G_i$ and the statistical characteristics of the random variable $\xi$ is established using a linear functional $J$ defined by this density as
\begin{equation}\label{eqFunctional}
  J \varphi = (g,\varphi)_{L_2(\Omega)} = \int_\Omega g(x) \varphi(x) dx = \mathrm{E} \varphi(\xi),
\end{equation}
where $\mathrm{E}$ means the mathematical expectation. In formula~\eqref{eqFunctional}, it is required to substitute the Legendre polynomials~\eqref{eqDefLeg} instead of a function $\varphi$, then
\[
  G_i = \mathrm{E} \hat P_i(\xi), \ \ \ i = 0,1,2,\dots
\]

The series representation and the corresponding approximation by a partial sum of this series, similar to~\eqref{eqFourier} and~\eqref{eqFourierTrunc}, can be obtained for an arbitrary orthonormal basis of $L_2(\Omega)$. But the Legendre polynomials have an important advantage; namely, expansion coefficients $G_i$ are easily expressed via initial moments of the random variable $\xi$:
\[
  \mathrm{M}_k = \mathrm{E} \xi^k, \ \ \ k = 1,2,\dots
\]

Indeed, using the formulae~\eqref{eqLegExp} and~\eqref{eqDefLeg}, we have
\begin{equation}\label{eqPreRandomGi}
  \begin{aligned}
  \mathrm{E} \hat P_i(\xi) & = \frac{1}{2^i} \sqrt{\frac{2i+1}{2}} \sum\limits_{k = 0}^{\lfloor i/2 \rfloor} \frac{(-1)^k(2i-2k)!}{k!(i-k)!(i-2k)!} \, \mathrm{E} \xi^{i-2k} \\
  & = \frac{1}{2^i} \sqrt{\frac{2i+1}{2}} \sum\limits_{k = 0}^{\lfloor i/2 \rfloor} \frac{(-1)^k(2i-2k)!}{k!(i-k)!(i-2k)!} \, \mathrm{M}_{i-2k},
  \end{aligned}
\end{equation}
where $\mathrm{M}_0 = 1$.

Another advantage of the chosen basis is that the antiderivative of any Legendre polynomial with index $i > 0$ is represented as a linear combination of two Legendre polynomials with neighboring indices due to recurrence relation~\eqref{eqNNLegRec2} (for $i = 0$, one of the polynomials in a linear combination is also $\hat P_0$). Thus, having a projection estimate of the density, we can find a projection estimate of the distribution function by the simplest operations. The expressions required for this are obtained below.

By integrating the left-hand and right-hand sides of the relation~\eqref{eqNNLegRec2} and taking into account the property of the Legendre polynomials~\cite{Sue_79}, namely the equality $P_{i+1}(-1) = P_{i-1}(-1)$, we have
\begin{equation}\label{eqRecAux}
  (2i+1) \int_{-1}^x P_i(\theta) d\theta = P_{i+1}(x) - P_{i-1}(x), \ \ \ i = 1,2,\dots,
\end{equation}
or
\begin{align*}
  & \sqrt{\frac{2i+1}{2}} \int_{-1}^x P_i(\theta) d\theta = \frac{1}{\sqrt{2(2i+1)}} \, P_{i+1}(x) - \frac{1}{\sqrt{2(2i+1)}} \, P_{i-1}(x) \\
  & \ \ \ = \frac{1}{\sqrt{(2i+1)(2i+3)}} \, \sqrt{\frac{2i+3}{2}} \, P_{i+1}(x) - \frac{1}{\sqrt{(2i-1)(2i+1)}} \, \sqrt{\frac{2i-1}{2}} \, P_{i-1}(x),
\end{align*}
i.e.,
\begin{equation}\label{eqNNLegRec3a}
  \int_{-1}^x \hat P_i(\theta) d\theta = \frac{1}{\sqrt{(2i+1)(2i+3)}} \, \hat P_{i+1}(x) - \frac{1}{\sqrt{(2i-1)(2i+1)}} \, \hat P_{i-1}(x), \ \ \ i = 1,2,\dots
\end{equation}

Given $i = 0$, we obtain
\[
  \int_{-1}^x P_0(\theta) d\theta = \int_{-1}^x d\theta = x + 1 = P_1(x) + P_0(x),
\]
or
\[
  \frac{1}{\sqrt{2}} \int_{-1}^x P_0(\theta) d\theta = \frac{1}{\sqrt{2}} \, P_1(x) + \frac{1}{\sqrt{2}} \, P_0(x) = \frac{1}{\sqrt{3}} \, \sqrt{\frac{3}{2}} \, P_1(x) + \frac{1}{\sqrt{2}} \, P_0(x),
\]
i.e.,
\begin{equation}\label{eqNNLegRec3b}
  \int_{-1}^x \hat P_0(\theta) d\theta = \frac{1}{\sqrt{3}} \, \hat P_1(x) + \hat P_0(x).
\end{equation}

Relationship~\eqref{eqDistrib} between functions $g$ and $f$, as well as recurrent relations~\eqref{eqNNLegRec3a} and~\eqref{eqNNLegRec3b}, implies that
\begin{equation}\label{eqPreRandomFi}
  F_0 = G_0 - \frac{G_1}{\sqrt{3}}, \ \ \ F_i = \frac{G_{i-1}}{\sqrt{(2i-1)(2i+1)}} - \frac{G_{i+1}}{\sqrt{(2i+1)(2i+3)}}, \ \ \ i = 1,2,\dots
\end{equation}

\begin{Remark}\label{remSegmentAB}
The condition $\Omega = [-1,1]$ does not reduce the generality of the above reasoning. Similar relations can be derived if the range of the random variable $\xi$ is $\Omega = [a,b]$ and $g,f \in L_2(\Omega)$. In this case, the orthonormal Legendre polynomials $\{\hat P_i^*\}_{i = 0}^\infty$ are defined by the expression
\[
  \hat P_i^*(x) = \sqrt{\frac{2i+1}{b-a}} \, P_i \biggl( \frac{2x-(b+a)}{b-a} \biggr).
\]
\end{Remark}

\section{Algorithms for Jointly Obtaining Randomized Projection Estimates of Density~and~Distribution Functions}\label{secMonteCarlo}

The randomization of the projection estimate of the density $g$ is obtained by the first of the formulae~\eqref{eqGFProj} as a result of calculating the linear functional~\eqref{eqFunctional} for the Legendre polynomials~\eqref{eqDefLeg} by the Monte Carlo method with independent realizations of the random variable $\xi$:
\begin{equation}\label{eqRandomGiBase}
  \bar{G}_i = \frac{1}{N} \sum\limits_{l = 1}^N \hat P_i(\xi_l), \ \ \ i = 0,1,2,\dots,
\end{equation}
where $\xi_l$ is the $l$th realization and $N$ is the sample size (number of realizations).

Estimates $\bar{G}_i$ can also be obtained from expression~\eqref{eqPreRandomGi} as
\begin{equation}\label{eqRandomGi}
  \bar{G}_i = \frac{1}{2^i} \sqrt{\frac{2i+1}{2}} \sum\limits_{k = 0}^{\lfloor i/2 \rfloor} \frac{(-1)^k(2i-2k)!}{k!(i-k)!(i-2k)!} \, \bar{\mathrm{M}}_{i-2k},
\end{equation}
where $\bar{\mathrm{M}}_k$ are statistical estimates of initial moments of the random variable $\xi$,
\begin{equation}\label{eqMomentEstimates}
  \bar{\mathrm{M}}_k = \frac{1}{N} \sum\limits_{l = 1}^N \xi_l^k, \ \ \ k = 1,2,\dots,
\end{equation}
and $\bar{\mathrm{M}}_0 = 1$. Using the formulae~\eqref{eqLegExpA} and~\eqref{eqLegExpB}, we can write the expressions for the estimates of the density expansion coefficients separately for even indices ($i = 2j$), as
\begin{equation}\label{eqRandomGiA}
  \bar{G}_{2j} = \frac{1}{2^{2j}} \sqrt{\frac{4j+1}{2}} \sum\limits_{k = 0}^j \frac{(-1)^k(4j-2k)!}{k!(2j-k)!(2j-2k)!} \, \bar{\mathrm{M}}_{2j-2k},
\end{equation}
and for odd indices ($i = 2j+1$), as
\begin{equation}\label{eqRandomGiB}
  \bar{G}_{2j+1} = \frac{1}{2^{2j+1}} \sqrt{\frac{4j+3}{2}} \sum\limits_{k = 0}^j \frac{(-1)^k(4j-2k+2)!}{k!(2j-k+1)!(2j-2k+1)!} \, \bar{\mathrm{M}}_{2j-2k+1}.
\end{equation}

The randomization of the projection estimate of the distribution function $f$ is obtained by the second of the formulae~\eqref{eqGFProj} and recurrence relations~\eqref{eqPreRandomFi}:
\begin{equation}\label{eqRandomFi}
  \bar{F}_0 = \bar{G}_0 - \frac{\bar{G}_1}{\sqrt{3}}, \ \ \ \bar{F}_i = \frac{\bar{G}_{i-1}}{\sqrt{(2i-1)(2i+1)}} - \frac{\bar{G}_{i+1}}{\sqrt{(2i+1)(2i+3)}}, \ \ \ i = 1,2,\dots
\end{equation}

Here and below, the dependence of the estimates of both expansion coefficients and functions on the sample size $N$ is not indicated for simplicity.

\begin{Remark}\label{remCoefGF}
To obtain a projection estimate of the distribution function $f$ based on the first $n+1$ Legendre polynomials, it is necessary to find estimates of expansion coefficients of the density $g$ with respect to the first $n+2$ Legendre polynomials.
\end{Remark}

Further, we present the first algorithm for jointly obtaining randomized projection estimates of the density and distribution function of the random variable $\xi$.

\begin{Algorithm}[Jointly obtaining projection estimates of the density and distribution function using explicit formulae for the Legendre polynomials]\label{algProjEstimate1}

~

0.\;Set the projection expansion length $n$ and the sample size $N$.

1.\;Simulate $N$ realizations $\xi_l$ of the random variable $\xi$, $l = 1,2,\dots,N$.

2.\;Find statistical estimates $\bar{\mathrm{M}}_k$ of initial moments of the random variable $\xi$ using the formula~\eqref{eqMomentEstimates}, $k = 1,2,\dots,n+1$. Set $\bar{\mathrm{M}}_0 = 1$.

3.\;Find estimates of expansion coefficients $\bar{G}_i$ of the density $g$ using the formula~\eqref{eqRandomGi} or the formulae~\eqref{eqRandomGiA} and~\eqref{eqRandomGiB}, $i = 0,1,\dots,n+1$.

4.\;Find estimates of expansion coefficients $\bar{F}_i$ of the distribution function $f$ using the formulae~\eqref{eqRandomFi}, $i = 0,1,\dots,n$.

5.\;Find randomized projection estimates $\bar{g}^{\langle n \rangle}$ and $\bar{f}^{\langle n \rangle}$ of the density and distribution function, respectively, using the formulae~\eqref{eqGFProj}.
\end{Algorithm}

In step 3 in Algorithm~\ref{algProjEstimate1}, errors can occur due to the peculiarities of machine arithmetic when calculating expansion coefficients $\bar{G}_i$. To avoid this, it is recommended to use the formula~\eqref{eqRandomGiBase} together with the recurrence relation~\eqref{eqNNLegRec1} and the expression~\eqref{eqDefLeg}.

Next, we formulate the second algorithm for jointly obtaining randomized projection estimates of the density and distribution function.

\begin{Algorithm}[Jointly obtaining projection estimates of the density and distribution function using the recurrence relation for the Legendre polynomials]\label{algProjEstimate2}

~

0.\;Set the projection expansion length $n$ and the sample size $N$.

1.\;Simulate $N$ realizations $\xi_l$ of the random variable $\xi$, $l = 1,2,\dots,N$.

2.\;Find estimates of expansion coefficients $\bar{G}_i$ of the density $g$ using the formula~\eqref{eqRandomGiBase} together with the recurrence relation~\eqref{eqNNLegRec1} and the expression~\eqref{eqDefLeg}, $i = 0,1,\dots,n+1$.

3.\;Find estimates of expansion coefficients $\bar{F}_i$ of the distribution function $f$ using the formulae~\eqref{eqRandomFi}, $i = 0,1,\dots,n$.

4.\;Find randomized projection estimates $\bar{g}^{\langle n \rangle}$ and $\bar{f}^{\langle n \rangle}$ of the density and distribution function, respectively, using the formulae~\eqref{eqGFProj}.
\end{Algorithm}

\section{Analysis and Conditional Optimization of Randomized Projection Estimates}\label{secOptim}

In this section, we analyze the error of the projection estimate $\bar{g}^{\langle n \rangle}$ relative to the density $g$ of the random variable $\xi$ in $L_2(\Omega)$. Let $B(g,\bar{g}^{\langle n \rangle}) = \mathrm{E} \| g - \bar{g}^{\langle n \rangle} \|_{L_2(\Omega)}$; then,
\[
  B^2(g,\bar{g}^{\langle n \rangle}) = \bigl( \mathrm{E} \| g - \bar{g}^{\langle n \rangle} \|_{L_2(\Omega)} \bigr)^2 \leqslant \mathrm{E} \int_\Omega \bigl( g(x) - \bar{g}^{\langle n \rangle}(x) \bigr)^2 dx
\]
due to Jensen's inequality~\cite{GihSco_77}. Further, we consider the following expression:
\[
  \int_\Omega \bigl( g(x) - \bar{g}^{\langle n \rangle}(x) \bigr)^2 dx = \| g - \bar{g}^{\langle n \rangle} \|_{L_2(\Omega)}^2 = \| g - g^{\langle n \rangle} + g^{\langle n \rangle} - \bar{g}^{\langle n \rangle} \|_{L_2(\Omega)}^2.
\]

Functions $g - g^{\langle n \rangle}$ and $g^{\langle n \rangle} - \bar{g}^{\langle n \rangle}$ are orthogonal in $L_2(\Omega)$ since they belong to different subspaces. The first one is formed by the Legendre polynomials $\hat P_i$ for $i = n+1,n+2,\dots$, and the second one is formed by the Legendre polynomials $\hat P_i$ for $i = 0,1,\dots,n$. This is a consequence of the equalities
\[
  g(x) = \sum\limits_{i=0}^\infty G_i \hat P_i(x), \ \ \ g^{\langle n \rangle}(x) = \sum\limits_{i=0}^n G_i \hat P_i(x), \ \ \ \bar{g}^{\langle n \rangle}(x) = \sum\limits_{i=0}^n \bar{G}_i \hat P_i(x);
\]
hence,
\[
  \| g - g^{\langle n \rangle} + g^{\langle n \rangle} - \bar{g}^{\langle n \rangle} \|_{L_2(\Omega)}^2 = \| g - g^{\langle n \rangle} \|_{L_2(\Omega)}^2 + \| g^{\langle n \rangle} - \bar{g}^{\langle n \rangle} \|_{L_2(\Omega)}^2.
\]

According to Parseval's identity~\cite{Sue_79}, we have
\[
  \| g^{\langle n \rangle} - \bar{g}^{\langle n \rangle} \|_{L_2(\Omega)}^2 = \biggl\| \sum\limits_{i=0}^n (G_i - \bar{G}_i) \hat P_i(x) \biggl\|_{L_2(\Omega)}^2 = \sum\limits_{i=0}^n (G_i - \bar{G}_i)^2
\]
and
\[
  \int_\Omega \bigl( g(x) - \bar{g}^{\langle n \rangle}(x) \bigr)^2 dx = \sum\limits_{i=0}^n (G_i - \bar{G}_i)^2 + \| g - g^{\langle n \rangle} \|_{L_2(\Omega)}^2.
\]

Therefore, taking into account that the unbiased estimate of mathematical expectation is used~\cite{Cra_99}, namely $\mathrm{E} \bar{G}_i = G_i$, we obtain
\begin{equation}\label{eqErrorG0}
  \begin{aligned}
    B^2(g,\bar{g}^{\langle n \rangle}) & \leqslant \mathrm{E} \biggl( \sum\limits_{i=0}^n (G_i - \bar{G}_i)^2 + \| g - g^{\langle n \rangle} \|_{L_2(\Omega)}^2 \biggr) \\
    & = \sum\limits_{i=0}^n \mathrm{E} (G_i - \bar{G}_i)^2 + \| g - g^{\langle n \rangle} \|_{L_2(\Omega)}^2 \\
    & = \sum\limits_{i=0}^n \mathrm{D} \bar{G}_i + \| g - g^{\langle n \rangle} \|_{L_2(\Omega)}^2,
  \end{aligned}
\end{equation}
where $\mathrm{D}$ denotes the variance.

The equality $\bar{G}_0 = G_0$ is satisfied since $\hat P_0(x) = \mathrm{const}$, i.e., $\mathrm{D} \bar{G}_0 = 0$. For $i > 0$, the variance of estimates $\bar{G}_i$ and the sample size $N$ are inversely proportional~\cite{MihVoi_06}; therefore, using additionally the inequality~\eqref{eqEstimateL2A} under the condition $g \in W_2^s(\Omega)$ (see Remark~\ref{remSreal}), we find
\begin{equation}\label{eqErrorG}
  B^2(g,\bar{g}^{\langle n \rangle}) \leqslant \frac{C_1^g n}{N} + \frac{C_2^g}{n^{2s}},
\end{equation}
where $C_1^g,C_2^g > 0$ are constants independent of $n$ and $N$.

The error of the projection estimate $\bar{f}^{\langle n \rangle}$ relative to the distribution function $f$ of the random variable $\xi$ in $L_2(\Omega)$ is analyzed similarly:
\begin{equation}\label{eqErrorF}
  B^2(f,\bar{f}^{\langle n \rangle}) \leqslant \frac{C_1^f n}{N} + \frac{C_2^f}{n^{2s+2}},
\end{equation}
where $C_1^f,C_2^f > 0$ are constants independent of $n$ and $N$.

From the obtained estimates, it is clear how the mean square error depends on the projection expansion length $n$ and the sample size $N$.

Further, we consider a problem of the optimal (consistent) choice of parameters for statistical algorithms to obtain projection estimates of the density and distribution function: the projection expansion length $n$ and the sample size $N$. For jointly obtaining these estimates by Algorithms~\ref{algProjEstimate1} and~\ref{algProjEstimate2}, it is sufficient to consider the optimal choice of parameters for density estimation only and use them for distribution function estimation. This is because the degree of smoothness of $f$ is greater than the degree of smoothness of $g$, so the sequence $\bar{f}^{\langle n \rangle}$ converges to $f$ faster than the sequence $\bar{g}^{\langle n \rangle}$ converges to $g$.

The main result is stated using the symbol ``$\asymp$''. The expression $u(n) \asymp v(n)$ for suitable functions $u$ and $v$ means that $u(n) = O(v(n))$ and $v(n) = O(u(n))$ for $n \to \infty$; i.e., there exist constants $C_*,C^* > 0$ such that $C_* |v(n)| \leqslant |u(n)| \leqslant C^* |v(n)|$ $\forall \, n > n^*$, where $n,n^*$ are natural numbers.

\begin{Theorem}\label{thmMain}
Let the density $g$ of the random variable $\xi$ belong to $W_2^s(\Omega)$, and let $\bar{g}^{\langle n \rangle}$ be the randomized projection estimate of the density obtained by Algorithms~\ref{algProjEstimate1} or~\ref{algProjEstimate2} with the projection expansion length $n$ and the sample size $N$. Then the minimum complexity of obtaining the estimate $\bar{g}^{\langle n \rangle}$ is achieved with parameters $n = n_\mathrm{opt}$ and $N = N_\mathrm{opt}$ that satisfy relations
\[
  N_\mathrm{opt} \asymp n_\mathrm{opt}^{2s+1}, \ \ \ N_\mathrm{opt} \asymp \gamma^{-(2s+1)/s}, \ \ \ n_\mathrm{opt} \asymp \gamma^{-1/s},
\]
where $\gamma$ is the required approximation accuracy for the density $g$ in $L_2(\Omega)$, i.e., $B(g,\bar{g}^{\langle n \rangle}) \leqslant \gamma$.
\end{Theorem}

\begin{proof}
To find the optimal parameters $n_\mathrm{opt}$ and $N_\mathrm{opt}$ for the estimate $\bar{g}^{\langle n \rangle}$, it is sufficient to equate terms~\cite{Mih_00} in the formula~\eqref{eqErrorG} and express the required approximation accuracy $\gamma$ from the relation
\[
  \frac{C_1^g n}{N} + \frac{C_2^g}{n^{2s}} = \gamma^2.
\]

From the equality
\[
  \frac{C_1^g n}{N} = \frac{C_2^g}{n^{2s}},
\]
we obtain the relationship for the optimal parameters, i.e., $N_\mathrm{opt} \asymp n_\mathrm{opt}^{2s+1}$, as well as expressions
\[
  \frac{C_1^g n}{N} = \frac{\gamma^2}{2} \ \ \ \text{and} \ \ \ \frac{C_2^g}{n^{2s}} = \frac{\gamma^2}{2},
\]
and consequently, $N_\mathrm{opt} \asymp \gamma^{-(2s+1)/s}$ and $n_\mathrm{opt} \asymp \gamma^{-1/s}$.
\end{proof}

Theorem~\ref{thmMain} establishes the relationship between parameters in Algorithms~\ref{algProjEstimate1} and~\ref{algProjEstimate2}: $n$~and $N$, as well as the dependence of approximation accuracy on parameter $s$. By choosing parameters in this way, we have
\[
  B(g,\bar{g}^{\langle n \rangle}) \leqslant c N^{-s/(2s+1)},
\]
where $c > 0$ is a constant independent of $n$ and $N$.

\begin{Remark}\label{remForThmMain1}
The error in~\eqref{eqErrorG} of the randomized projection estimate $\bar{g}^{\langle n \rangle}$ relative to the density $g$ is based on the inequality~\eqref{eqEstimateL2A}, but the inequality~\eqref{eqEstimateL2B} can also be used. Then, taking into account Remark~\ref{remSreal}, we have
\[
  B^2(g,\bar{g}^{\langle n \rangle}) \leqslant \frac{C_1^g n}{N} + C_2^g \, \frac{\Gamma(n-s+2)}{\Gamma(n+s+2)}.
\]

It is possible to formulate an analogue of Theorem~\ref{thmMain} and show that the following relationship between parameters for estimating the density $g$ is conditionally optimal:
\[
  N_\mathrm{opt} \asymp n_\mathrm{opt} \, \frac{\Gamma(n_\mathrm{opt}+s+2)}{\Gamma(n_\mathrm{opt}-s+2)}.
\]
\end{Remark}

\begin{Remark}\label{remForThmMain2}
If the distribution function $f$ of the random variable $\xi$ belongs to $W_2^{s+1}(\Omega)$ (this condition holds if the density $g$ of the random variable $\xi$ belongs to $W_2^s(\Omega)$, i.e., if conditions of Theorem~\ref{thmMain} are satisfied), then the minimum complexity of obtaining randomized projection estimate $\bar{f}^{\langle n \rangle}$ by Algorithms~\ref{algProjEstimate1} or~\ref{algProjEstimate2} is achieved with parameters $n = n_\mathrm{opt}$ and $N = N_\mathrm{opt}$ that satisfy the relations
\[
  N_\mathrm{opt} \asymp n_\mathrm{opt}^{2s+3}, \ \ \ N_\mathrm{opt} \asymp \gamma^{-(2s+3)/(s+1)}, \ \ \ n_\mathrm{opt} \asymp \gamma^{-1/(s+1)},
\]
where $\gamma$ is the required approximation accuracy for the function $f$ in $L_2(\Omega)$, i.e., $B(f,\bar{f}^{\langle n \rangle})  \leqslant \gamma$.

The proof is similar to the proof of Theorem~\ref{thmMain}. Another relationship between parameters can be found if we take the inequality~\eqref{eqEstimateL2B} instead of the inequality~\eqref{eqEstimateL2A} (see Remark~\ref{remForThmMain1}).
\end{Remark}

\section{Two-parameter Family of Densities with Different Degrees of Smoothness}\label{secTest}

In this section, a special example is proposed to test statistical algorithms for obtaining projection estimates depending on the projection expansion length, sample size, and smoothness of the estimated function.

\subsection{Densities with Different Degrees of Smoothness and Related Distribution Functions}

Let $\xi$ be the random variable defined by the density $g_{\nu_1,\nu_2}$:
\begin{equation}\label{eqDensn1n2}
  g_{\nu_1,\nu_2}(x) = \gamma_{\nu_1,\nu_2} \left\{ \begin{array}{ll}
    1-(-x)^{\nu_1} & \text{for} ~ x \in [-1,0), \\
    1-x^{\nu_2} & \text{for} ~ x \in [0,1], \\
    0 & \text{otherwise},
  \end{array} \right.
\end{equation}
where $\nu_1$ and $\nu_2$ are parameters (natural numbers), and $\gamma_{\nu_1,\nu_2}$ is a normalizing constant, with
\[
  \frac{1}{\gamma_{\nu_1,\nu_2}} = \int_{-1}^0 \bigl( 1-(-x)^{\nu_1} \bigr) dx + \int_0^1 (1-x^{\nu_2}) dx = \frac{\nu_1}{\nu_1+1} + \frac{\nu_2}{\nu_2+1},
\]
because
\[
  \int_{-1}^0 \bigl( 1-(-x)^\nu \bigr) dx = \int_0^1 (1-x^\nu) dx = \biggl( x - \frac{x^{\nu+1}}{\nu+1} \biggr) \bigg|_0^1 = 1 - \frac{1}{\nu+1} = \frac{\nu}{\nu+1}.
\]

Further, we assume that $\nu_1 \neq \nu_2$. The function $g_{\nu_1,\nu_2}$ has the following properties:

(a)\;$g_{\nu_1,\nu_2}$ is continuous on the set of real numbers;

(b)\;the support of $g_{\nu_1,\nu_2}$ is the set $\Omega = [-1,1]$;

(c)\;the normalization condition holds:
\[
  \int_{-\infty}^{+\infty} g_{\nu_1,\nu_2}(x) dx = \int_{-1}^1 g_{\nu_1,\nu_2}(x) dx = 1;
\]

(d)\;$g_{\nu_1,\nu_2}$ is differentiable at $x = 0$ only $r$ times:
\[
  g_{\nu_1,\nu_2} \in C^r(\Omega), \ \ \ r = \min\{\nu_1,\nu_2\} - 1;
\]

(e)\;the derivative $g_{\nu_1,\nu_2}^{(r+1)}$ exists almost everywhere on $\Omega$; if the derivative is understood in a generalized sense, then
\[
  g_{\nu_1,\nu_2}^{(s)} \in L_2(\Omega), \ \ \ g_{\nu_1,\nu_2} \in W_2^s(\Omega), \ \ \ s = r+1 = \min\{\nu_1,\nu_2\}.
\]

Next, we determine the distribution function $f_{\nu_1,\nu_2}$ for the random variable $\xi$ with density~\eqref{eqDensn1n2}:
\begin{equation}\label{eqDistribEx}
  f_{\nu_1,\nu_2}(x) = \int_{-\infty}^x g_{\nu_1,\nu_2}(\theta) d\theta.
\end{equation}

It is easy to see that $f_{\nu_1,\nu_2}(x) = 0$ for $x < -1$ and $f_{\nu_1,\nu_2}(x) = 1$ for $x > 1$. Moreover,
\begin{gather*}
  \int_{-1}^x \bigl( 1-(-\theta)^\nu \bigr) d\theta = \biggl( \theta + (-1)^{\nu+1} \frac{\theta^{\nu+1}}{\nu+1} \biggr) \bigg|_{-1}^x = \frac{\nu}{\nu+1} + x + \frac{(-x)^{\nu+1}}{\nu+1}, \ \ \ x \in [-1,0], \\
  \int_0^x (1-\theta^\nu) d\theta = \biggl( \theta - \frac{\theta^{\nu+1}}{\nu+1} \biggr) \bigg|_0^x = x - \frac{x^{\nu+1}}{\nu+1}, \ \ \ x \in [0,1],
\end{gather*}
and consequently,
\begin{equation}\label{eqDistribn1n2}
  f_{\nu_1,\nu_2}(x) = \left\{ \begin{array}{ll}
    \displaystyle 0\vphantom{\frac12} & \text{for} ~ x < -1, \\
    \displaystyle \gamma_{\nu_1,\nu_2} \biggl( \frac{\nu_1}{\nu_1+1} + x + \frac{(-x)^{\nu_1+1}}{\nu_1+1} \biggr) & \text{for} ~ x \in [-1,0), \\
    \displaystyle \gamma_{\nu_1,\nu_2} \biggl( \frac{\nu_1}{\nu_1+1} + x - \frac{x^{\nu_2+1}}{\nu_2+1} \biggr) & \text{for} ~ x \in [0,1], \\
    \displaystyle 1\vphantom{\frac12} & \text{for} ~ x > 1.
  \end{array} \right.
\end{equation}

The function $f_{\nu_1,\nu_2}$ is differentiable $r+1$ times at $x = 0$ since $f_{\nu_1,\nu_2}^{(r+1)} = g_{\nu_1,\nu_2}^{(r)}$ due to the relationship~\eqref{eqDistribEx}. Thus,
\[
  f_{\nu_1,\nu_2} \in C^{r+1}(\Omega), \ \ \ f_{\nu_1,\nu_2}^{(s+1)} \in L_2(\Omega), \ \ \ f_{\nu_1,\nu_2} \in W_2^{s+1}(\Omega).
\]

If we do not restrict ourselves to the space $W_2^s(\Omega)$ with natural $s$ (see Remark~\ref{remSreal}), then we can show that $s < \min\{\nu_1,\nu_2\} + 1/2$ and $\sup s = \min\{\nu_1,\nu_2\} + 1/2$.

Consider the case $\nu_1 > \nu_2$. Then the generalized derivative $g_{\nu_1,\nu_2}^{(\nu_2)}$ is represented as a linear combination of functions $\eta(x)$ and $x^{\nu_1 - \nu_2} \eta(-x)$, where $\eta$ is the indicator  of $(0,\infty)$. It suffices to find a condition for parameter $\sigma$ which ensures that $\eta \in W_2^\sigma(\Omega)$, where $W_2^\sigma(\Omega)$ is defined by the formula~\eqref{eqDefWs} and $\sigma \in (0,1)$.

If $x,y$ have the same sign, then $|\eta(x) - \eta(y)| = 0$, and if $x$ and $y$ have different signs, then $|\eta(x) - \eta(y)| = 1$. Hence,
\[
  \int_{\Omega^2} \frac{|\eta(x)-\eta(y)|^2}{|x-y|^{1+2\sigma}} \, dx dy = \int_{-1}^0 \biggl[ \int_0^1 \frac{dy}{|x-y|^{1+2\sigma}} \biggr] dx + \int_0^1 \biggl[ \int_{-1}^0 \frac{dy}{|x-y|^{1+2\sigma}} \biggr] dx.
\]

The integrals on the right-hand side of the latter equality coincide since the integrand does not change when the signs of $x$ and $y$ change simultaneously. The convergence condition for them is $\sigma < 1/2$. Indeed, for $\sigma < 1/2$, we have
\begin{align*}
  & \int_{-1}^0 \biggl[ \int_0^1 \frac{dy}{|x-y|^{1+2\sigma}} \biggr] dx = \int_{-1}^0 \biggl[ \int_0^1 \frac{dy}{(y-x)^{1+2\sigma}} \biggr] dx = -\frac{1}{2\sigma} \int_{-1}^0 \frac{1}{(y-x)^{2\sigma}} \bigg|_0^1 dx \\
  & \ \ \ = -\frac{1}{2\sigma} \int_{-1}^0 \biggl[ \frac{1}{(1-x)^{2\sigma}} - \frac{1}{(-x)^{2\sigma}} \biggr] dx = \frac{1}{2\sigma(2\sigma - 1)} \biggl[ -\frac{1}{(1-x)^{2\sigma-1}} \bigg|_{-1}^0 + \frac{1}{(-x)^{2\sigma-1}} \bigg|_{-1}^0 \biggr] \\
  & \ \ \ = \frac{1}{2\sigma(2\sigma-1)} \biggl[ \frac{1}{2^{2\sigma-1}} - 2 \biggr] = \frac{2^{-2\sigma} - 1}{\sigma(2\sigma-1)}
\end{align*}
and
\[
  \int_{\Omega^2} \frac{|\eta(x)-\eta(y)|^2}{|x-y|^{1+2\sigma}} \, dx dy = \frac{2(2^{-2\sigma} - 1)}{\sigma(2\sigma-1)}.
\]

For $\sigma \geqslant 1/2$, these integrals obviously diverge. Therefore, $g_{\nu_1,\nu_2}^{(\nu_2)} \in W_2^\sigma(\Omega)$ for $\sigma < 1/2$ and $g_{\nu_1,\nu_2} \in W_2^s(\Omega)$ for $s < \nu_2+1/2$.

The case $\nu_1 < \nu_2$ is similar to the considered case, so finally $g_{\nu_1,\nu_2} \in W_2^s(\Omega)$ and $f_{\nu_1,\nu_2} \in W_2^{s+1}(\Omega)$ provided $s < \min\{\nu_1,\nu_2\} + 1/2$ and $\sup s = \min\{\nu_1,\nu_2\} + 1/2$.

\subsection{Modeling Random Variables with Given Test Distributions Using Monte Carlo Method}

The modeling formula for the random variable $\xi$ with distribution function $f_{\nu_1,\nu_2}$ for parameters $\nu_1$ and $\nu_2$ can be derived using the inverse function method~\cite{MihVoi_06}: $\xi = f_{\nu_1,\nu_2}^{-1}(\alpha)$, where $f_{\nu_1,\nu_2}^{-1}$ is the inverse function to $f_{\nu_1,\nu_2}$, and $\alpha$ is the random variable having a uniform distribution on $(0,1)$.

Given the distribution function $f_{\nu_1,\nu_2}$, we can obtain the algorithm for modeling the random variable $\xi$.

\begin{Algorithm}[Modeling the random variable with given test density and distribution function]\label{algInverseFunc}

~

0.\;Set parameters $\nu_1$, $\nu_2$, calculate $\gamma_{\nu_1,\nu_2}$ and $f_{\nu_1,\nu_2}(0)$:
\[
 \gamma_{\nu_1,\nu_2}= \biggl( \frac{\nu_1}{\nu_1+1} + \frac{\nu_2}{\nu_2+1} \biggr)^{-1}, \ \ \
  f_{\nu_1,\nu_2}(0) = \gamma_{\nu_1,\nu_2} \, \frac{\nu_1}{\nu_1+1} =  \frac{\nu_1 (\nu_2+1)}{\nu_1 + 2 \nu_1 \nu_2 + \nu_2}.
\]

1.\;Obtain a realization of the random variable $\alpha$ having a uniform distribution on $(0,1)$.

2.\;If $\alpha < f_{\nu_1,\nu_2}(0)$, then $\xi$ is a root of the algebraic equation
\begin{equation}\label{eqXiModeling1}
  \frac{\nu_1}{\nu_1+1} + x + \frac{(-x)^{\nu_1+1}}{\nu_1+1} - \frac{\alpha}{\gamma_{\nu_1,\nu_2}} = 0
\end{equation}
from $(-1,0)$, otherwise $\xi$ is a root of the algebraic equation
\begin{equation}\label{eqXiModeling2}
  \frac{\nu_1}{\nu_1+1} + x - \frac{x^{\nu_2+1}}{\nu_2+1} - \frac{\alpha}{\gamma_{\nu_1,\nu_2}} = 0
\end{equation}
from $[0,1)$.
\end{Algorithm}

For small $\nu_1$ and $\nu_2$, roots can be found analytically. Next, we obtain the modeling formulae for two cases used below.

\begin{Proposition}\label{exXiModeling1}
For the random variable $\xi $ with density
\[
  g(x) = \frac{6}{7} \left\{ \begin{array}{ll}
    1+x & \text{for} ~ x \in [-1,0), \\
    1-x^{2} & \text{for} ~ x \in [0,1], \\
    0 & \text{otherwise},
  \end{array} \right.
\]
the modeling formula is as follows:
\begin{equation}\label{eqXiModelingEx1}
  \xi = \left\{ \begin{array}{ll}
    \displaystyle \sqrt{\frac{7\alpha}{3}} - 1 & \text{for} ~ \displaystyle \alpha < \frac{3}{7}, \\ [-2ex] \\
    \displaystyle -\mathop{\mathrm{Re}} A + \sqrt{3} \, \mathop{\mathrm{Im}} A & \text{for} ~ \displaystyle \alpha \geqslant \frac{3}{7},
  \end{array} \right.
\end{equation}
where
\[
  A = \sqrt[3]{\frac{3-7\alpha}{4} + \sqrt{ \frac{(7\alpha-3)^2}{16} - 1}}.
\]
\end{Proposition}

\begin{proof}
The density $g$ of the random variable $\xi$ is included in a two-parameter family~\eqref{eqDensn1n2} for $\nu_1 = \nu$, $\nu_2 = \nu+1$, and $\nu = 1$: $g = g_{1,2}$. For given parameters, we have
\[
  \frac{\nu_1}{\nu_1+1} = \frac{1}{2}, \ \ \ \gamma_{1,2} = \frac{6}{7}, \ \ \ f_{1,2}(0) = \frac{3}{7}.
\]

First, we consider the case $\alpha < f_{1,2}(0) = 3/7$; i.e., we should solve the equation~\eqref{eqXiModeling1}. This is the quadratic equation
\[
  \frac{1}{2} + x + \frac{x^2}{2} - \frac{7\alpha}{6} = 0, \ \ \ \text{or} \ \ \ x^2 + bx + c = 0,
\]
where $b = 2$, $c = 1 - 7\alpha/3 > 0$.

The function $v(x) = x^2 + 2x + 1 - 7\alpha/3$ has a minimum at $x^* = -1$ since $v'(x^*) = 2x^* + 2 = 0$, $v''(x^*) = 2 > 0$, and $v(x^*) = -7\alpha/3 < 0$. This means that the quadratic equation has two real roots (the discriminant is positive) and the largest of them determines~$\xi$:
\[
  \xi = \max\{x_1,x_2\} > x^*, \ \ \ x_{1,2} = -\frac{b}{2} \pm \sqrt{ \frac{b^2}{4} - c } = -1 \pm \sqrt{\frac{7\alpha}{3}}, \ \ \ \xi = \sqrt{\frac{7\alpha}{3}} - 1 \in (-1,0).
\]

Second, we consider the case $\alpha \geqslant f_{1,2}(0) = 3/7$; i.e., we should solve the equation~\eqref{eqXiModeling2}. This is the cubic equation
\[
  \frac{1}{2} + x - \frac{x^3}{3} - \frac{7\alpha}{6} = 0, \ \ \ \text{or} \ \ \ x^3 + px + q = 0,
\]
where $p = -3$, $q = (7\alpha-3)/2 > 0$.

The function $w(x) = x^3 - 3x + (7\alpha-3)/2$ has extrema at $x_{1,2}^* = \pm 1$: $w'(x^*) = 3(x_{1,2}^*)^2 - 3 = 0$. A minimum is reached at $x_1^* = 1$: $w''(x_1^*) = 6 x_1^* = 6 > 0$ and $w(x_1^*) = 7(\alpha-1)/2 < 0$. A maximum is reached at $x_2^* = -1$: $w''(x_2^*) = 6 x_2^* = -6 < 0$ and $w(x_2^*) = (1 + 7\alpha)/2 > 0$. This means that the cubic equation has three real roots (the discriminant is positive), and $\xi$ is determined by the root that lies between the smallest and the largest roots. By using Cardano's formulae for roots~\cite{KornKorn_00}, we have
\[
  x_1 = A+B, \ \ \ x_{2,3} = -\frac{A+B}{2} \pm \frac{A-B}{2} \, \sqrt{3} \, \mathrm{i},
\]
where
\[
  A = \sqrt[3]{-\frac{q}{2} + \sqrt{Q}}, \ \ \ B = \sqrt[3]{-\frac{q}{2} - \sqrt{Q}}, \ \ \
  Q = \biggl( \frac{p}{3} \biggr)^3 + \biggl( \frac{q}{2} \biggr)^2,
\]
and $Q < 0$ which corresponds to the positive discriminant; therefore, $A$ and $B$ are complex conjugate numbers:
\[
  \frac{A+B}{2} = \mathop{\mathrm{Re}} A, \ \ \ \frac{A-B}{2 \, \mathrm{i}} = \mathop{\mathrm{Im}} A
  \ \ \ \text{and} \ \ \
  x_1 = 2 \mathop{\mathrm{Re}} A, \ \ \ x_{2,3} = -\mathop{\mathrm{Re}} A \mp \sqrt{3} \, \mathop{\mathrm{Im}} A.
\]

Let $z = -q/2 + \sqrt{Q}$. Since $\mathop{\mathrm{Re}} z = -q/2 < 0$ and $\mathop{\mathrm{Im}} z = \sqrt{-Q} > 0$, we obtain $\arg z \in (\pi/2,\pi)$; therefore, $\arg A \in (\pi/6,\pi/3)$ and
\[
  \mathop{\mathrm{Re}} A, \mathop{\mathrm{Im}} A > 0, \ \ \ \tan \frac{\pi}{6} < \frac{\mathop{\mathrm{Im}} A}{\mathop{\mathrm{Re}} A} < \tan \frac{\pi}{3},
  \ \ \ \text{or} \ \ \
  \mathop{\mathrm{Re}} A < \sqrt{3} \, \mathop{\mathrm{Im}} A < 3 \mathop{\mathrm{Re}} A.
\]

Thus,
\[
  x_2 < x_3 < x_1, \ \ \ \xi = -\frac{A+B}{2} - \frac{A-B}{2} \, \sqrt{3} \, \mathrm{i} = -\mathop{\mathrm{Re}} A + \sqrt{3} \, \mathop{\mathrm{Im}} A \in (0,1),
\]
i.e., the formula~\eqref{eqXiModelingEx1} is valid.
\end{proof}

\begin{Proposition}\label{exXiModeling2}
For the random variable $\xi $ with density
\[
  g(x) = \frac{12}{17} \left\{ \begin{array}{ll}
    1+x^{3} & \text{for} ~ x \in [-1,0), \\
    1-x^{2} & \text{for} ~ x \in [0,1], \\
    0 & \text{otherwise},
  \end{array} \right.
\]
the modeling formula is as follows:
\begin{equation}\label{eqXiModelingEx2}
  \xi = \left\{ \begin{array}{ll}
    \displaystyle -\sqrt{\frac{Y}{2}} + \sqrt{-\frac{Y}{2} + \sqrt{\frac{2}{Y}}} & \text{for} ~ \displaystyle \alpha < \frac{9}{17}, \\ [-2ex] \\
    \displaystyle -\mathop{\mathrm{Re}} A + \sqrt{3} \, \mathop{\mathrm{Im}} A & \text{for} ~ \displaystyle \alpha \geqslant \frac{9}{17},
  \end{array} \right.
\end{equation}
where
\[
  Y = \omega - \frac{17\alpha-9}{9 \omega}, \ \ \ \omega = \sqrt[3]{1 + \sqrt{1 + \frac{(17\alpha-9)^3}{729}}}, \ \ \ A = \sqrt[3]{\frac{9-17\alpha}{8} + \sqrt{\frac{(17\alpha-9)^2}{64} - 1}}.
\]
\end{Proposition}

\begin{proof}
The density $g$ of the random variable $\xi$ is included in a two-parameter family~\eqref{eqDensn1n2} for $\nu_1 = \nu+1$, $\nu_2 = \nu$, and $\nu = 2$: $g = g_{3,2}$. In this case,
\[
  \frac{\nu_1}{\nu_1+1} = \frac{3}{4}, \ \ \ \gamma_{3,2} = \frac{12}{17}, \ \ \ f_{3,2}(0) = \frac{9}{17}.
\]

The proof is the same as for Proposition~\ref{exXiModeling1}, so some details are omitted. We only note that for $\alpha < f_{3,2}(0) = 9/17$, the equation~\eqref{eqXiModeling1} is the quartic equation
\[
  \frac{3}{4} + x + \frac{x^4}{4} - \frac{17\alpha}{12} = 0, \ \ \ \text{or} \ \ \ x^4 + bx + c = 0,
\]
where $b = 4$, $c = 3 - 17\alpha/3 > 0$, and the polynomial $v(x) = x^4 + 4x + 3 - 17\alpha/3$ has a minimum at $x^* = -1$, and $v(x^*) = -17\alpha/3 < 0$. The quartic equation has two real roots, and the largest of them determines $\xi$. It is convenient to find roots using Ferrari's formulae~\cite{KornKorn_00}.

For $\alpha \geqslant f_{3,2}(0) = 9/17$, the equation~\eqref{eqXiModeling2} is the cubic equation that is solved similarly to the cubic equation from the proof of Proposition~\ref{exXiModeling1}. Such reasoning leads to the formula~\eqref{eqXiModelingEx2}.
\end{proof}

\subsection{Expansion Coefficients of Test Functions (Fourier--Legendre Series)}

To exactly calculate the second term of the projection estimate error~\eqref{eqErrorG0} in examples, we find expansion coefficients $G_{\nu_1,\nu_2,i}$ of the density $g_{\nu_1,\nu_2}$, as well as expansion coefficients $F_{\nu_1,\nu_2,i}$ of the distribution function $f_{\nu_1,\nu_2}$ with respect to Legendre polynomials~\eqref{eqDefLeg}:
\[
  g_{\nu_1,\nu_2}(x) = \sum\limits_{i=0}^\infty G_{\nu_1,\nu_2,i} \hat P_i(x), \ \ \
  f_{\nu_1,\nu_2}(x) = \sum\limits_{i=0}^\infty F_{\nu_1,\nu_2,i} \hat P_i(x).
\]

First, we obtain the following values:
\begin{equation}\label{eqDefQpm}
  Q_{\nu,i}^- = \int_{-1}^0 x^\nu P_i(x) dx, \ \ \ Q_{\nu,i}^+ = \int_0^1 x^\nu P_i(x) dx.
\end{equation}

We multiply the left-hand and right-hand sides of the relation~\eqref{eqNNLegRec1} by $x^{\nu-1}$ and integrate over the interval $[-1,0]$:
\[
  (i+1) \int_{-1}^0 x^{\nu-1} P_{i+1}(x) dx = (2i+1) \int_{-1}^0 x^\nu P_i(x) dx - i \int_{-1}^0 x^{\nu-1} P_{i-1}(x) dx,
\]
or
\begin{equation}\label{eqQpmRec1}
  (i+1) Q_{\nu-1,i+1}^- = (2i+1) Q_{\nu,i}^- - i Q_{\nu-1,i-1}^-, \ \ \ Q_{\nu,i}^- = \frac{(i+1) Q_{\nu-1,i+1}^- + i Q_{\nu-1,i-1}^-}{2i+1}.
\end{equation}

Similarly, by multiplying the left-hand and right-hand sides of the relation~\eqref{eqNNLegRec1} by $x^{\nu-1}$ and integrating over the interval $[0,1]$, we obtain
\[
  Q_{\nu,i}^+ = \frac{(i+1) Q_{\nu-1,i+1}^+ + i Q_{\nu-1,i-1}^+}{2i+1}.
\]

These relations can be formally applied for $i = 0$ when $Q_{\nu,-1}^- = Q_{\nu,-1}^+ = 0$ but not for $\nu = 0$. Therefore, we should consider the case $\nu = 0$ separately:
\[
  Q_{0,i}^- = \int_{-1}^0 P_i(x) dx, \ \ \ Q_{0,i}^+ = \int_0^1 P_i(x) dx.
\]

For $i = 0$ and $i = 1$, we have
\[
  Q_{0,0}^- = \int_{-1}^0 P_0(x) dx = \int_{-1}^0 dx = 1, \ \ \ Q_{0,1}^- = \int_{-1}^0 P_1(x) dx = \int_{-1}^0 x dx = -\frac{1}{2},
\]
and then we use the relation~\eqref{eqRecAux}:
\[
  Q_{0,i}^- = \int_{-1}^0 P_i(x) dx = \frac{P_{i+1}(0) - P_{i-1}(0)}{2i+1}, \ \ \ i = 1,2,\dots
\]

If $i \neq 0$ is even, then $P_{i+1}(0) = P_{i-1}(0) = 0$ according to the formula~\eqref{eqLegExpB}, and $Q_{0,i}^- = 0$. If $i$ is odd, then we can apply the explicit formula~\eqref{eqLegExpA} or obtain an additional recurrence relation. We choose the latter way and take into account the relation~\eqref{eqNNLegRec1} for $x = 0$:
\[
  (i+1) P_{i+1}(0) = -i P_{i-1}(0).
\]

Then
\begin{align*}
  \frac{Q_{0,i}^-}{Q_{0,i-2}^-} & = \frac{P_{i+1}(0) - P_{i-1}(0)}{2i+1} \, \frac{2i-3}{P_{i-1}(0) - P_{i-3}(0)} \\
  & = \frac{-\frac{i}{i+1} \, P_{i-1}(0) - P_{i-1}(0)}{2i+1} \, \frac{2i-3}{P_{i-1}(0) + \frac{i-1}{i-2} \, P_{i-1}(0)} \\
  & = \frac{-\frac{i}{i+1} - 1}{2i+1} \, \frac{2i-3}{1 + \frac{i-1}{i-2}},
\end{align*}
i.e.,
\[
  Q_{0,i}^- = \frac{2-i}{i+1} \, Q_{0,i-2}^-, \ \ \ i = 3,5,7,\dots
\]

The same reasoning leads to the following results:
\[
  Q_{0,0}^+ = \int_0^1 P_0(x) dx = \int_0^1 dx = 1, \ \ \ Q_{0,1}^+ = \int_0^1 P_1(x) dx = \int_0^1 x dx = \frac{1}{2},
\]
and
\[
  Q_{0,i}^+ = \frac{P_{i-1}(0) - P_{i+1}(0)}{2i+1}, \ \ \ i = 1,2,\dots,
\]
where $Q_{0,i}^+ = 0$ for even $i \neq 0$. For odd $i$, we have
\[
  Q_{0,i}^+ = \frac{2-i}{i+1} \, Q_{0,i-2}^+, \ \ \ i = 3,5,7,\dots
\]

Thus, we obtain the general expression for calculating $Q_{\nu,i}^-$ and $Q_{\nu,i}^+$ with arbitrary non-negative integers $\nu$ and $i$:
\[
  Q_{\nu,i}^\pm = \left\{ \begin{array}{ll}
    \displaystyle 1\vphantom{\frac12} & \text{for} ~ \nu = i = 0, \\
    \displaystyle \pm \frac{1}{2} & \text{for} ~ \nu = 0 ~ \text{and} ~ i = 1, \\
    \displaystyle 0\vphantom{\frac12} & \text{for} ~ \nu = 0 ~ \text{and} ~ i = 2,4,6,\dots, \\
    \displaystyle \frac{2-i}{i+1} \, Q_{0,i-2}^\pm & \text{for} ~ \nu = 0 ~ \text{and} ~ i = 3,5,7,\dots, \\
    \displaystyle \frac{(i+1) Q_{\nu-1,i+1}^\pm + i Q_{\nu-1,i-1}^\pm}{2i+1} & \text{otherwise},
  \end{array} \right.
\]
so that expansion coefficients of the density $g_{\nu_1,\nu_2}$ with respect to Legendre polynomials~\eqref{eqDefLeg} are expressed as follows (the relation~\eqref{eqDefQpm} is also used here):
\begin{align*}
  G_{\nu_1,\nu_2,i} & = \int_{-1}^1 g_{\nu_1,\nu_2}(x) \hat P_i(x) dx \\
  & = \gamma_{\nu_1,\nu_2} \, \sqrt{\frac{2i+1}{2}} \biggl( \int_{-1}^0 \bigl( 1-(-x)^{\nu_1} \bigr) P_i(x) dx + \int_0^1 (1-x^{\nu_2}) P_i(x) dx \biggr) \\
  & = \gamma_{\nu_1,\nu_2} \, \sqrt{\frac{2i+1}{2}} \, \bigl( Q_{0,i}^- - (-1)^{\nu_1} Q_{\nu_1,i}^- + Q_{0,i}^+ - Q_{\nu_2,i}^+ \bigr).
\end{align*}

Expressions for expansion coefficients of the distribution function $f_{\nu_1,\nu_2}$ are similarly obtained:
\begin{align*}
  F_{\nu_1,\nu_2,i} & = \int_{-1}^1 f_{\nu_1,\nu_2}(x) \hat P_i(x) dx \\
  & = \gamma_{\nu_1,\nu_2} \, \sqrt{\frac{2i+1}{2}} \biggl( \int_{-1}^0 \biggl[ \frac{\nu_1}{\nu_1+1} + x + \frac{(-x)^{\nu_1+1}}{\nu_1+1} \biggr] P_i(x) dx \\
  & \ \ \ \ \ \ + \int_0^1 \biggl[ \frac{\nu_1}{\nu_1+1} + x - \frac{x^{\nu_2+1}}{\nu_2+1} \biggr] P_i(x) dx \biggr) \\
  & = \gamma_{\nu_1,\nu_2} \, \sqrt{\frac{2i+1}{2}} \, \biggl( \frac{\nu_1}{\nu_1+1} \, Q_{0,i}^- + Q_{1,i}^- + \frac{(-1)^{\nu_1+1}}{\nu_1+1} \, Q_{\nu_1+1,i}^- \\
  & \ \ \ \ \ \ + \frac{\nu_1}{\nu_1+1} \, Q_{0,i}^+ + Q_{1,i}^+ - \frac{1}{\nu_2+1} \, Q_{\nu_2+1,i}^+ \biggr).
\end{align*}

These expressions for expansion coefficients of the density  $g_{\nu_1,\nu_2}$ and distribution function $f_{\nu_1,\nu_2}$ are used for their approximation as
\[
  g_{\nu_1,\nu_2}(x) \approx g_{\nu_1,\nu_2}^{\langle n \rangle}(x) = \sum\limits_{i=0}^n G_{\nu_1,\nu_2,i} \hat P_i(x), \ \ \
  f_{\nu_1,\nu_2}(x) \approx f_{\nu_1,\nu_2}^{\langle n \rangle}(x) = \sum\limits_{i=0}^n F_{\nu_1,\nu_2,i} \hat P_i(x),
\]
and for the approximation error:
\begin{equation}\label{eqTestEstimates}
  \begin{aligned}
    \varepsilon_g^{\langle n \rangle} & = \| g_{\nu_1,\nu_2} - g_{\nu_1,\nu_2}^{\langle n \rangle} \|_{L_2(\Omega)} = \biggl( \int_{-1}^1 g_{\nu_1,\nu_2}^2(x) dx - \sum\limits_{i=0}^n G_{\nu_1,\nu_2,i}^2 \biggr)^{1/2}, \\
    \varepsilon_f^{\langle n \rangle} & = \| f_{\nu_1,\nu_2} - f_{\nu_1,\nu_2}^{\langle n \rangle} \|_{L_2(\Omega)} = \biggl( \int_{-1}^1 f_{\nu_1,\nu_2}^2(x) dx - \sum\limits_{i=0}^n F_{\nu_1,\nu_2,i}^2 \biggr)^{1/2}.
  \end{aligned}
\end{equation}

\subsection{Analysis of Convergence Rate for Expansions of Test Functions}

Consider the function
\[
  y_\nu^-(x) = \left\{ \begin{array}{ll}
    x^\nu & \text{for} ~ x \in [-1,0), \\
    0 & \text{for} ~ x \in [0,1],
  \end{array} \right.
\]
where $\nu$ is a parameter (natural number). Its expansion coefficients $Y_{\nu,i}^-$ with respect to Legendre polynomials~\eqref{eqDefLeg} are expressed through the previously found values $Q_{\nu,i}^-$:
\[
  y_\nu^-(x) = \sum\limits_{i=0}^\infty Y_{\nu,i}^- \hat P_i(x), \ \ \ Y_{\nu,i}^- = \sqrt{\frac{2i+1}{2}} \, Q_{\nu,i}^-.
\]

Further, we derive the recurrence relation for $Q_{\nu,i}^-$, different from the relation~\eqref{eqQpmRec1}. We multiply the left-hand and right-hand sides of the relation~\eqref{eqNNLegRec1} by $x^\nu$ and integrate over the interval $[-1,0]$:
\[
  (i+1) \int_{-1}^0 x^\nu P_{i+1}(x) dx = (2i+1) \int_{-1}^0 x^{\nu+1} P_i(x) dx - i \int_{-1}^0 x^\nu P_{i-1}(x) dx.
\]

Next, we use the rule of integration by parts:
\[
  \int_{-1}^0 x^{\nu+1} P_i(x) dx = x^{\nu+1} \Theta_i(x) \bigg|_{-1}^0 - (\nu+1) \int_{-1}^0 x^\nu \Theta_i(x) dx, \ \ \ \Theta_i(x) = \int_{-1}^x P_i(\theta) d\theta,
\]
and consequently, taking into account the equality~\eqref{eqRecAux}, we obtain
\[
  \int_{-1}^0 x^{\nu+1} P_i(x) dx = -\frac{\nu+1}{2i+1} \int_{-1}^0 x^\nu \bigl( P_{i+1}(x) - P_{i-1}(x) \bigr) dx;
\]
therefore,
\[
  (i+1) \int_{-1}^0 x^\nu P_{i+1}(x) dx = -(\nu+1) \int_{-1}^0 x^\nu \bigl( P_{i+1}(x) - P_{i-1}(x) \bigr) dx - i \int_{-1}^0 x^\nu P_{i-1}(x) dx,
\]
or
\[
  (i+1) Q_{\nu,i+1}^- = -(\nu+1) (Q_{\nu,i+1}^- - Q_{\nu,i-1}^-) - i Q_{\nu,i-1}^-,
\]
i.e.,
\[
  Q_{\nu,i+1}^- = \frac{\nu-i+1}{\nu+i+2} \, Q_{\nu,i-1}^- \ \ \ \text{and} \ \ \ Y_{\nu,i+1}^- = \sqrt{\frac{2i+3}{2i-1}} \, \frac{\nu-i+1}{\nu+i+2} \, Y_{\nu,i-1}^-.
\]

The series formed by the squared expansion coefficients $Y_{\nu,i}^-$ converges since $y_\nu^- \in L_2(\Omega)$. It can be represented as a sum of two series:
\[
  \sum\limits_{i=0}^\infty (Y_{\nu,i}^-)^2 = \sum\limits_{j=0}^\infty (Y_{\nu,2j}^-)^2 + \sum\limits_{j=0}^\infty (Y_{\nu,2j+1}^-)^2.
\]

The Raabe--Duhamel test~\cite{Fih_06} implies that the first series (similarly, for the second one) converges in the same way as the series
\[
  \sum\limits_{j=1}^\infty \frac{1}{j^{2\nu+2}},
\]
since the sequence
\[
  R_j = j \biggl( \frac{(Y_{\nu,2j}^-)^2}{(Y_{\nu,2j+2}^-)^2} - 1 \biggr) = j \biggl( \frac{4j-1}{4j+3} \biggl[ \frac{\nu+2j+2}{\nu-2j+1} \biggr]^2 - 1 \biggr)
\]
has the limit $\lim\limits_{j \to \infty} R_j = 2\nu+2$, but the convergence of this series is equivalent to the convergence of the integral
\[
  \int_1^{+\infty} \frac{dt}{t^{2\nu+2}} = -\frac{1}{(2\nu+1) \, t^{2\nu+1}} \bigg|_1^{+\infty} = -\frac{1}{2\nu+1} \biggl( \lim\limits_{t \to \infty} \frac{1}{t^{2\nu+1}} - 1 \biggr)
\]
which takes place under the condition $2\nu+1 > 0$, or $\nu > -1/2$.

As a result, using Parseval's identity, we find
\[
  \| y_\nu^- - y_\nu^{-\langle n \rangle} \|_{L_2(\Omega)}^2 = \sum\limits_{i=n+1}^\infty (Y_{\nu,i}^-)^2 \leqslant \frac{C^2}{n^{2\nu+1}}
  \ \ \ \text{and} \ \ \
  \| y_\nu^- - y_\nu^{-\langle n \rangle} \|_{L_2(\Omega)} \leqslant \frac{C}{n^{\nu+1/2}}, \ \ \ C > 0,
\]
where
\[
  y_\nu^{-\langle n \rangle}(x) = \sum\limits_{i=0}^n Y_{\nu,i}^- \hat P_i(x),
\]
and this corresponds to the estimate~\eqref{eqEstimateL2A} with the limit value $s = \nu+1/2$ (see Remark~\ref{remSreal}).

The obtained result can be transferred to the function
\[
  y_\nu^+(x) = \left\{ \begin{array}{ll}
    0 & \text{for} ~ x \in [-1,0), \\
    x^\nu & \text{for} ~ x \in [0,1],
  \end{array} \right.
\]
and its expansion coefficients $Y_{\nu,i}^+$ with respect to Legendre polynomials~\eqref{eqDefLeg}. The easiest way to prove this is the use of the equality $|Y_{\nu,i}^-| = |Y_{\nu,i}^+|$, which follows from the relation between expansion coefficients of an arbitrary function from $L_2(\Omega)$ and its reflection~\cite{Ryb_Comp25}. Therefore, the same result holds for the function $g_{\nu_1,\nu_2}$. This result can be extended to the function $f_{\nu_1,\nu_2}$ due to its degree of smoothness is greater by one than the degree of smoothness of $g_{\nu_1,\nu_2}$.

Thus,
\begin{equation}\label{eqTestConvergenceRate}
  \varepsilon_g^{\langle n \rangle} \approx \frac{C^g}{n^{\min\{\nu_1,\nu_2\} + 1/2}}, \ \ \ \varepsilon_f^{\langle n \rangle} \approx \frac{C^f}{n^{\min\{\nu_1,\nu_2\} + 3/2}},
\end{equation}
where $C^g,C^f > 0$ are some constants. Moreover, we can assume that $\nu_1$ and $\nu_2$ are real positive numbers, and if we consider $g_{\nu_1,\nu_2}$ not to be a density but some function not bound by the probability theoretical framework, then the condition $\nu_1,\nu_2 > -1/2$ is admissible. Such a convergence reflects that $g_{\nu_1,\nu_2}^{(s)} \in W_2^s(\Omega)$ and $f_{\nu_1,\nu_2}^{(s)} \in W_2^{s+1}(\Omega)$ subject to $s < \min\{\nu_1,\nu_2\} + 1/2$ and $\sup s = \min\{\nu_1,\nu_2\} + 1/2$. It corresponds to the estimate~\eqref{eqEstimateL2A}.

\section{Numerical Experiments}\label{secNumerics}

In this section, we present the results of the joint estimation of the density and distribution function for two examples that use a two-parameter family~\eqref{eqDensn1n2} of densities with different degrees of smoothness. In these examples, the results are presented in the tables that contain errors of projection estimates of the density and distribution function in $L_2(\Omega)$. We study the dependence of error on the projection expansion length (for the maximum degree $n$ of the Legendre polynomials, values 4, 8, 16, 32, and 64 are used, i.e., $n = 2^{k+2}$ for $k = 0,1,\dots,4$), on the sample size $N = 2^{m+9}$ for $m = 0,1,\dots,18$, and on the degree of smoothness of the approximated density (see Examples~\ref{ex1} and~\ref{ex2} below). Algorithm~\ref{algProjEstimate2} is applied for estimation.

These examples show the approximation errors $\varepsilon_g^{\langle n \rangle}$ and $\varepsilon_f^{\langle n \rangle}$, which are calculated using the formulae~\eqref{eqTestEstimates}, i.e., deterministic components of projection estimate errors. In the tables, they are in rows marked by the symbol ``$*$''. The remaining rows contain errors that include deterministic and stochastic components. The formulae for these errors follow from Parseval's identity:
\begin{align*}
  \bar\varepsilon_g^{\langle n \rangle} & = \biggl( \int_{-1}^1 g_{\nu_1,\nu_2}^2(x) dx - \sum\limits_{i=0}^n G_{\nu_1,\nu_2,i}^2 + \sum\limits_{i=0}^n (G_{\nu_1,\nu_2,i} - \bar{G}_{\nu_1,\nu_2,i})^2 \biggr)^{1/2}, \\
  \bar\varepsilon_f^{\langle n \rangle} & = \biggl( \int_{-1}^1 f_{\nu_1,\nu_2}^2(x) dx - \sum\limits_{i=0}^n F_{\nu_1,\nu_2,i}^2 + \sum\limits_{i=0}^n (F_{\nu_1,\nu_2,i} - \bar{F}_{\nu_1,\nu_2,i})^2 \biggr)^{1/2},
\end{align*}
where $\bar{G}_{\nu_1,\nu_2,i}$ and $\bar{F}_{\nu_1,\nu_2,i}$ are estimates of expansion coefficients $G_{\nu_1,\nu_2,i}$ and $F_{\nu_1,\nu_2,i}$, respectively. For an arbitrary density $g$, a similar formula was used to obtain the estimate~\eqref{eqErrorG0}.

\begin{Example}\label{ex1}
Let $g_{\nu_1,\nu_2}$ be the density from a two-parameter family~\eqref{eqDensn1n2}, and let $f_{\nu_1,\nu_2}$ be the related distribution function described by the formula~\eqref{eqDistribn1n2}, with the following parameters: $\nu_1 = \nu$, $\nu_2 = \nu + 1$ for $\nu = 1$. The modeling formula is given in Proposition~\ref{exXiModeling1}.

The function $g_{1,2}$ is continuous, non-differentiable at $x = 0$, but its first-order derivative exists almost everywhere on $\Omega$: $g_{1,2} \in C(\Omega)$, $g_{1,2}^{(1)} \in L_2(\Omega)$, i.e., $g_{1,2} \in W_2^1(\Omega)$. However, if we do not restrict ourselves to the space $W_2^s(\Omega)$ with natural $s$ (see Remark~\ref{remSreal}), then $g_{1,2} \in W_2^s(\Omega)$ for $s < \min\{\nu_1,\nu_2\} + 1/2 = 3/2$ ($\sup s = 3/2$). The degree of smoothness of $f_{1,2}$ is greater, and $f_{1,2} \in W_2^{s+1}(\Omega)$. This corresponds to the formulae~\eqref{eqTestConvergenceRate} which take the form
\[
  \varepsilon_g^{\langle n \rangle} \approx \frac{C^g}{n^{3/2}}, \ \ \ \varepsilon_f^{\langle n \rangle} \approx \frac{C^f}{n^{5/2}}
\]
for given parameters $\nu_1$ and $\nu_2$.

According to Theorem~\ref{thmMain}, for the limit value $s = 3/2$, to achieve the required approximation accuracy $B(g_{1,2},\bar{g}_{1,2}^{\langle n \rangle}) = \gamma \leqslant c N^{-3/8}$, $c > 0$, the conditionally optimal parameters should be as follows: $N_\mathrm{opt} \asymp n_\mathrm{opt}^4$; this is confirmed by the statistical modeling results (see Tables~\ref{tabErrorG1} and~\ref{tabErrorF1}). In the row ``$*$'' of Table~\ref{tabErrorG1}, the deterministic component of error decreases by approximately $2^{3/2} \approx 2.8$ times (in Table~\ref{tabErrorF1}, by approximately $2^{5/2} \approx 5.6$ times) when the projection expansion length $n$ is doubled. In the rest of this table, errors corresponding to optimal parameters $n_\mathrm{opt}$ and $N_\mathrm{opt}$ are highlighted in bold, and they are consistent with the relationship $N_\mathrm{opt} \asymp n_\mathrm{opt}^4$. Table~\ref{tabErrorF1} demonstrates higher accuracy of distribution function estimation.

In our calculations with the formulae for errors, the following values are used (squared norms of the functions $g_{1,2}$ and $f_{1,2}$):
\[
  \int_{-1}^1 g_{1,2}^2(x) dx = \frac{156}{245}, \ \ \ \int_{-1}^1 f_{1,2}^2(x) dx = \frac{235}{343}.
\]

For clarity, Figure~\ref{picTables} contains the approximation errors in graphical form. One axis shows $k = 0,1,\dots,4$, which determines the projection expansion length, $n = 2^{k+2}$, and another axis shows $m = 0,1,\dots,18$, which determines the sample size, $N = 2^{m+9}$. The vertical axis corresponds to the density approximation errors.

The example under consideration corresponds to the left part of Figure~\ref{picTables} which presents two surfaces. The first one (\textcolor{red}{red}) corresponds to the obtained computational error, it is formed on data from Table~\ref{tabErrorG1}. The second one (\textcolor{blue}{blue}, with marked nodes) corresponds to the theoretical error according to the formula~\eqref{eqErrorG} with $s = 3/2$:
\[
  \epsilon_g^{\langle n \rangle} \approx \sqrt{\frac{C_1^g n}{N} + \frac{C_2^g}{n^3}}.
\]

Constants $C_1^g = 0.885$ and $C_2^g = 0.276$ are approximately determined based on the condition of a minimum sum of squared deviations between the theoretical error and computational error.
\end{Example}

\begin{table}[p]
\begin{center}
\renewcommand{\arraystretch}{1.05}
\caption{The approximation errors $\varepsilon_g^{\langle n \rangle}$ \textcolor{blue}{(in the row ``$*$'')} and $\bar\varepsilon_g^{\langle n \rangle}$ \textcolor{blue}{(in the remaining rows)} for sample sizes $N = 2^{m+9}$ (Example~\ref{ex1})}\label{tabErrorG1}
\begin{tabular}{|c|c|c|c|c|c|}
  \hline
  $m$ &\ \ \ \ $n = 4$\ \ \ \ &\ \ \ \ $n = 8$\ \ \ \ &\ \ \ $n = 16$\ \ \ &\ \ \ $n = 32$\ \ \ &\ \ \ $n = 64$\ \ \ \\
  \hline
  \hline
  $*$ & 0.024614 & 0.009800 & 0.003838 & 0.001442 & 0.000527 \\
  \hline
  \hline
   0 & \textbf{0.065300} & 0.068625 & 0.119784 & 0.168044 & 0.209996 \\
   1 & 0.029999 & 0.020867 & 0.064303 & 0.107275 & 0.144097 \\
   2 & 0.026360 & 0.031022 & 0.046008 & 0.061732 & 0.093907 \\
   3 & 0.035117 & 0.028862 & 0.036685 & 0.045302 & 0.064636 \\
   4 & 0.028603 & \textbf{0.020236} & 0.026896 & 0.035032 & 0.050834 \\
   5 & 0.025879 & 0.015305 & 0.018405 & 0.025655 & 0.034847 \\
   6 & 0.025867 & 0.014401 & 0.014068 & 0.017283 & 0.023956 \\
   7 & 0.026084 & 0.015807 & 0.014762 & 0.017245 & 0.021088 \\
   8 & 0.024814 & 0.011465 & \textbf{0.008550} & 0.010850 & 0.015997 \\
   9 & 0.024640 & 0.010035 & 0.005534 & 0.006907 & 0.010124 \\
  10 & 0.024638 & 0.009976 & 0.004925 & 0.005121 & 0.007505 \\
  11 & 0.024622 & 0.009888 & 0.004352 & 0.003211 & 0.005116 \\
  12 & 0.024616 & 0.009815 & 0.004112 & \textbf{0.002461} & 0.003081 \\
  13 & 0.024616 & 0.009814 & 0.003937 & 0.002064 & 0.002367 \\
  14 & 0.024616 & 0.009807 & 0.003881 & 0.001813 & 0.001689 \\
  15 & 0.024615 & 0.009802 & 0.003852 & 0.001623 & 0.001308 \\
  16 & 0.024615 & 0.009803 & 0.003851 & 0.001553 & \textbf{0.001093} \\
  17 & 0.024615 & 0.009802 & 0.003852 & 0.001503 & 0.000830 \\
  18 & 0.024615 & 0.009801 & 0.003845 & 0.001470 & 0.000647 \\
  \hline
\end{tabular}
\end{center}
\end{table}

\begin{table}[p]
\begin{center}
\renewcommand{\arraystretch}{1.05}
\caption{The approximation errors $\varepsilon_f^{\langle n \rangle}$ \textcolor{blue}{(in the row ``$*$'')} and $\bar\varepsilon_f^{\langle n \rangle}$ \textcolor{blue}{(in the remaining rows)} for sample sizes $N = 2^{m+9}$ (Example~\ref{ex1})}\label{tabErrorF1}
\begin{tabular}{|c|c|c|c|c|c|}
  \hline
  $m$ & \ \ \ \ $n = 4$\ \ \ \ &\ \ \ \ $n = 8$\ \ \ \ &\ \ \ $n = 16$\ \ \ &\ \ \ $n = 32$\ \ \ &\ \ \ $n = 64$\ \ \ \\
  \hline
  \hline
  $*$ & 0.005772 & 0.001069 & 0.000198 & 0.000036 & $6.48 \cdot 10^{-6}$ \\ 
  \hline
  \hline
   0 & 0.027351 & 0.026935 & 0.027769 & 0.028143 & 0.028244 \\
   1 & 0.010817 & 0.009287 & 0.010189 & 0.010621 & 0.010828 \\
   2 & 0.006060 & 0.004198 & 0.004744 & 0.005001 & 0.005172 \\
   3 & 0.008635 & 0.007038 & 0.007179 & 0.007255 & 0.007313 \\
   4 & 0.006667 & 0.004018 & 0.004166 & 0.004234 & 0.004284 \\
   5 & 0.006149 & 0.002746 & 0.002680 & 0.002754 & 0.002788 \\
   6 & 0.006159 & 0.002666 & 0.002519 & 0.002544 & 0.002561 \\
   7 & 0.006186 & 0.002722 & 0.002571 & 0.002594 & 0.002604 \\
   8 & 0.005845 & 0.001519 & 0.001160 & 0.001182 & 0.001205 \\
   9 & 0.005784 & 0.001152 & 0.000535 & 0.000539 & 0.000558 \\
  10 & 0.005801 & 0.001230 & 0.000661 & 0.000647 & 0.000656 \\
  11 & 0.005783 & 0.001140 & 0.000449 & 0.000411 & 0.000420 \\
  12 & 0.005773 & 0.001081 & 0.000263 & 0.000183 & 0.000186 \\
  13 & 0.005775 & 0.001087 & 0.000282 & 0.000210 & 0.000209 \\
  14 & 0.005773 & 0.001080 & 0.000250 & 0.000161 & 0.000158 \\
  15 & 0.005772 & 0.001070 & 0.000202 & 0.000060 & 0.000051 \\
  16 & 0.005772 & 0.001072 & 0.000212 & 0.000086 & 0.000079 \\
  17 & 0.005772 & 0.001071 & 0.000208 & 0.000075 & 0.000067 \\
  18 & 0.005772 & 0.001070 & 0.000200 & 0.000046 & 0.000030 \\
  \hline
\end{tabular}
\end{center}
\end{table}

\begin{Example}\label{ex2}
Let $g_{\nu_1,\nu_2}$ be the density from a two-parameter family~\eqref{eqDensn1n2}, and let $f_{\nu_1,\nu_2}$ be the related distribution function described by the formula~\eqref{eqDistribn1n2}, with the following parameters: $\nu_1 = \nu+1$, $\nu_2 = \nu$ for $\nu = 2$. The modeling formula is given in Proposition~\ref{exXiModeling2}.

The function $g_{3,2}$ is continuous, differentiable everywhere on $\Omega$, and its second-order derivative exists almost everywhere on $\Omega$: $g_{3,2} \in C^1(\Omega)$, $g_{3,2}^{(2)} \in L_2(\Omega)$, hence $g_{3,2} \in W_2^2(\Omega)$. Considering the space $W_2^s(\Omega)$ with real non-negative $s$ (see Remark~\ref{remSreal}), we affirm that $g_{3,2} \in W_2^s(\Omega)$ for $s < \min\{\nu_1,\nu_2\} + 1/2 = 5/2$ ($\sup s = 5/2$). The degree of smoothness of $f_{3,2}$ is greater, and $f_{3,2} \in W_2^{s+1}(\Omega)$. This corresponds to the formulae~\eqref{eqTestConvergenceRate} when substituting given parameters $\nu_1$ and $\nu_2$:
\[
  \varepsilon_g^{\langle n \rangle} \approx \frac{C^g}{n^{5/2}}, \ \ \ \varepsilon_f^{\langle n \rangle} \approx \frac{C^f}{n^{7/2}}.
\]

Theorem~\ref{thmMain} implies that in order to achieve the required approximation accuracy $B(g_{3,2},\bar{g}_{3,2}^{\langle n \rangle}) = \gamma \leqslant c N^{-5/12}$, $c > 0$, with the limit value $s = 5/2$, the conditionally optimal parameters should satisfy the relationship $N_\mathrm{opt} \asymp n_\mathrm{opt}^6$; this is illustrated by the statistical modeling results from Tables~\ref{tabErrorG2} and~\ref{tabErrorF2}. In the row ``$*$'' of Table~\ref{tabErrorG2}, if the projection expansion length $n$ is doubled, then the deterministic component of error decreases by approximately $2^{5/2} \approx 5.6$ times (in Table~\ref{tabErrorF2}, by approximately $2^{7/2} \approx 11.2$ times). In the rest of this table, errors corresponding to optimal parameters $n_\mathrm{opt}$ and $N_\mathrm{opt}$ are shown in bold, and they are consistent with the relationship $N_\mathrm{opt} \asymp n_\mathrm{opt}^6$. Table~\ref{tabErrorF2} shows higher accuracy of the distribution function estimation.

In our calculations with the formulae for errors, the following values are used (squared norms of the functions $g_{3,2}$ and $f_{3,2}$):
\[
  \int_{-1}^1 g_{3,2}^2(x) dx = \frac{5928}{10115}, \ \ \ \int_{-1}^1 f_{3,2}^2(x) dx = \frac{7801}{10115}.
\]

Figure~\ref{picTables} contains a graphical representation of the numerical experiment. The meaning of different axes is described in Example~\ref{ex1}. This example corresponds to the right part of Figure~\ref{picTables} with two surfaces. The first one (\textcolor{red}{red}) is constructed from the obtained computational error given in Table~\ref{tabErrorG2}, and the second one (\textcolor{blue}{blue}, with marked nodes) corresponds to the theoretical error according to the formula~\eqref{eqErrorG} with $s = 5/2$:
\[
  \epsilon_g^{\langle n \rangle} \approx \sqrt{\frac{C_1^g n}{N} + \frac{C_2^g}{n^5}},
\]
where constants $C_1^g = 0.890$ and $C_2^g = 0.545$ are chosen from the same condition as in Example~\ref{ex1}.
\end{Example}

\begin{table}[p]
\begin{center}
\renewcommand{\arraystretch}{1.05}
\caption{The approximation errors $\varepsilon_g^{\langle n \rangle}$ \textcolor{blue}{(in the row ``$*$'')} and $\bar\varepsilon_g^{\langle n \rangle}$ \textcolor{blue}{(in the remaining rows)} for sample sizes $N = 2^{m+9}$ (Example~\ref{ex2})}\label{tabErrorG2}
\begin{tabular}{|c|c|c|c|c|c|}
  \hline
  $m$ & \ \ \ \ $n = 4$\ \ \ \ &\ \ \ \ $n = 8$\ \ \ \ &\ \ \ $n = 16$\ \ \ &\ \ \ $n = 32$\ \ \ &\ \ \ $n = 64$\ \ \ \\
  \hline
  \hline
  $*$ & 0.009757 & 0.001782 & 0.000327 & 0.000059 & 0.000011 \\
  \hline
  \hline
   0 & 0.095213 & 0.096279 & 0.110003 & 0.163112 & 0.202505 \\
   1 & 0.069075 & 0.072595 & 0.077403 & 0.101639 & 0.156190 \\
   2 & 0.037189 & 0.042272 & 0.055580 & 0.075076 & 0.109453 \\
   3 & \textbf{0.019912} & 0.020148 & 0.028046 & 0.044567 & 0.077312 \\
   4 & 0.015445 & 0.015796 & 0.018872 & 0.035310 & 0.052812 \\
   5 & 0.010799 & 0.007478 & 0.014730 & 0.021591 & 0.037072 \\
   6 & 0.011292 & 0.007615 & 0.010409 & 0.013876 & 0.026574 \\
   7 & 0.012527 & 0.009047 & 0.011426 & 0.013697 & 0.021402 \\
   8 & 0.010221 & 0.004889 & 0.006611 & 0.008514 & 0.014083 \\
   9 & 0.009919 & \textbf{0.002619} & 0.003359 & 0.004879 & 0.008398 \\
  10 & 0.009851 & 0.002347 & 0.003053 & 0.004197 & 0.005689 \\
  11 & 0.009779 & 0.002255 & 0.001818 & 0.002358 & 0.003913 \\
  12 & 0.009783 & 0.002329 & 0.001788 & 0.001999 & 0.003051 \\
  13 & 0.009766 & 0.001968 & 0.001211 & 0.001580 & 0.002321 \\
  14 & 0.009769 & 0.001911 & 0.000960 & 0.001176 & 0.001615 \\
  15 & 0.009761 & 0.001839 & \textbf{0.000798} & 0.000778 & 0.001129 \\
  16 & 0.009757 & 0.001793 & 0.000472 & 0.000455 & 0.000766 \\
  17 & 0.009757 & 0.001789 & 0.000409 & 0.000425 & 0.000553 \\
  18 & 0.009758 & 0.001789 & 0.000382 & 0.000298 & 0.000388 \\
  \hline
\end{tabular}
\end{center}
\end{table}

\begin{table}[p]
\begin{center}
\renewcommand{\arraystretch}{1.05}
\caption{The approximation errors $\varepsilon_f^{\langle n \rangle}$ \textcolor{blue}{(in the row ``$*$'')} and $\bar\varepsilon_f^{\langle n \rangle}$ \textcolor{blue}{(in the remaining rows)} for sample sizes $N = 2^{m+9}$ (Example~\ref{ex2})}\label{tabErrorF2}
\begin{tabular}{|c|c|c|c|c|c|}
  \hline
  $m$ & \ \ \ \ $n = 4$\ \ \ \ &\ \ \ \ $n = 8$\ \ \ \ &\ \ \ $n = 16$\ \ \ &\ \ \ $n = 32$\ \ \ &\ \ \ $n = 64$\ \ \ \\
  \hline
  \hline
  $*$ & 0.002779 & 0.000139 & 0.000014 & $1.40 \cdot 10^{-6}$ & $1.35 \cdot 10^{-7}$ \\ 
  \hline
  \hline
   0 & 0.050042 & 0.050077 & 0.050178 & 0.050338 & 0.050393 \\
   1 & 0.032134 & 0.032449 & 0.032474 & 0.032568 & 0.032637 \\
   2 & 0.020694 & 0.020649 & 0.020779 & 0.020854 & 0.020905 \\
   3 & 0.010729 & 0.010445 & 0.010513 & 0.010577 & 0.010651 \\
   4 & 0.007291 & 0.006854 & 0.006890 & 0.006980 & 0.007017 \\
   5 & 0.003914 & 0.002828 & 0.002946 & 0.003002 & 0.003052 \\
   6 & 0.003984 & 0.002933 & 0.002982 & 0.002999 & 0.003031 \\
   7 & 0.004923 & 0.004130 & 0.004160 & 0.004169 & 0.004182 \\
   8 & 0.003304 & 0.001842 & 0.001872 & 0.001879 & 0.001893 \\
   9 & 0.002904 & 0.000866 & 0.000874 & 0.000883 & 0.000893 \\
  10 & 0.002803 & 0.000414 & 0.000439 & 0.000454 & 0.000460 \\
  11 & 0.002790 & 0.000322 & 0.000301 & 0.000306 & 0.000311 \\
  12 & 0.002792 & 0.000341 & 0.000320 & 0.000322 & 0.000325 \\
  13 & 0.002782 & 0.000208 & 0.000167 & 0.000171 & 0.000174 \\
  14 & 0.002783 & 0.000215 & 0.000169 & 0.000170 & 0.000171 \\
  15 & 0.002780 & 0.000171 & 0.000108 & 0.000107 & 0.000108 \\
  16 & 0.002779 & 0.000147 & 0.000053 & 0.000052 & 0.000053 \\
  17 & 0.002779 & 0.000144 & 0.000043 & 0.000043 & 0.000043 \\
  18 & 0.002780 & 0.000158 & 0.000076 & 0.000075 & 0.000075 \\
  \hline
\end{tabular}
\end{center}
\end{table}

\begin{figure}[ht]
  \centering
  \includegraphics[bb = 0 0 562 283, scale = 0.8]{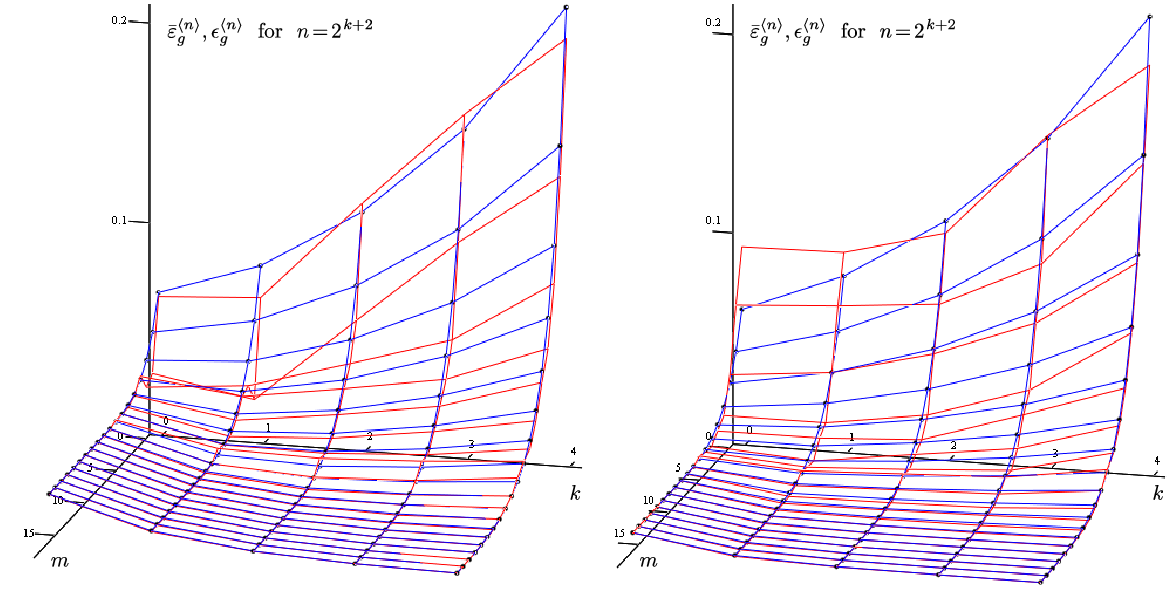}
  \vskip 2ex
  \caption{The approximation errors for densities (the left part for $g_{1,2}$ from Example~\ref{ex1}, the right part for $g_{3,2}$ from Example~\ref{ex2})}\label{picTables}
\end{figure}

\section{Comparison of Projection Density Estimate and Histogram}\label{secHistogram}

The classical method of estimating the density of a random variable is associated with a histogram~\cite{Cra_99}, which is very often used in applied problems. We consider a histogram as a projection estimate since the main results of this paper are related to projection estimates.

We can define block pulse functions on the set $\Omega$ as
\begin{equation}\label{eqDefBIF}
  \begin{gathered}
    \hat\Pi_i(x) = \frac{1}{\sqrt{h}} \left\{ \begin{array}{ll}
      1 & \text{for} ~ x \in [i h - 1,(i+1) h - 1), \\
      0 & \text{for} ~ x \notin [i h - 1,(i+1) h - 1),
    \end{array} \right. \\
    i = 0,1,\dots,L-1,
  \end{gathered}
\end{equation}
where $L$, a natural number, is the number of block pulse functions, and $h = 2/L$. It is advisable to redefine the function $\hat\Pi_{L-1}$ in such a way that it becomes continuous on the left at $x = 1$.

Block pulse functions~\eqref{eqDefBIF} form an orthonormal system of functions in $L_2(\Omega)$. This system is not complete, but it can be used to approximate an arbitrary function $u \in L_2(\Omega)$:
\[
  u(x) \approx u^{\langle L \rangle*}(x) = \sum\limits_{i=0}^{L-1} U_i^* \hat\Pi_i(x),
\]
where
\[
  U_i^* = (u,\hat\Pi_i)_{L_2(\Omega)} = \int_\Omega u(x) \hat\Pi_i(x) dx, \ \ \ i = 0,1,\dots,L-1.
\]

For $L = 2^m$ with natural $m$, the first $L$ Walsh or Haar functions on $\Omega$ are exactly expressed through block pulse functions~\eqref{eqDefBIF}. Therefore, the results of this section can easily be adapted to projection estimates of the distribution of a random variable using Walsh or Haar functions that form complete orthonormal systems of functions~\cite{GolEfiSkv_87}.

The approximation accuracy in $L_2(\Omega)$ is usually estimated as follows~\cite{MarAgo_81}:
\begin{equation}\label{eqEstimateL2C}
  \| u - u^{\langle L \rangle*} \|_{L_2(\Omega)} \leqslant \frac{C^*}{L},
\end{equation}
where $C^* > 0$ does not depend on $L$, under the condition $u \in W_2^1(\Omega)$.

As a function $u$ in the given formulae, we can use the density $g$ and the distribution function $f$. We restrict ourselves to the density $g$ only (corresponding expansion coefficients are further denoted by $G_i^*$):
\begin{equation}\label{eqGFApproxHist}
  g^{\langle L \rangle*}(x) = \sum\limits_{i=0}^{L-1} G_i^* \hat\Pi_i(x),
\end{equation}
where
\begin{align*}
  G_i^* & = \int_\Omega g(x) \hat\Pi_i(x) dx = \frac{1}{\sqrt{h}} \int_{i h - 1}^{(i+1) h - 1} g(x) dx \\
  & = \frac{f((i+1) h - 1) - f(i h - 1)}{\sqrt{h}}, \ \ \ i = 0,1,\dots,L-1.
\end{align*}

Thus, the calculation of expansion coefficients $G_{\nu_1,\nu_2,i}^*$ of the density $g_{\nu_1,\nu_2}$ from a two-parameter family~\eqref{eqDensn1n2} is reduced to the calculation of values of the corresponding distribution function $f_{\nu_1,\nu_2}$ described by the formula~\eqref{eqDistribn1n2}, i.e.,
\[
  G_{\nu_1,\nu_2,i}^* = \frac{f_{\nu_1,\nu_2}((i+1) h - 1) - f_{\nu_1,\nu_2}(i h - 1)}{\sqrt{h}}, \ \ \ i = 0,1,\dots,L-1.
\]

To calculate the approximation error for the density $g_{\nu_1,\nu_2}$, we can use the formula similar to the first of the formulae~\eqref{eqTestEstimates}:
\begin{equation}\label{eqTestEstimatesHist}
  \varepsilon_g^{\langle L \rangle*} = \| g_{\nu_1,\nu_2} - g_{\nu_1,\nu_2}^{\langle L \rangle*} \|_{L_2(\Omega)} = \biggl( \int_{-1}^1 g_{\nu_1,\nu_2}^2(x) dx - \sum\limits_{i=0}^{L-1} (G_{\nu_1,\nu_2,i}^*)^2 \biggr)^{1/2}.
\end{equation}

The histogram can be defined by the expression based on the approximation~\eqref{eqGFApproxHist}:
\[
  \bar{g}^{\langle L \rangle*}(x) = \sum\limits_{i=0}^{L-1} \bar{G}_i^* \hat\Pi_i(x),
\]
where $\bar{G}_i^*$ are estimates of expansion coefficients $G_i$ based on observations of the random variable $\xi$, $i = 0,1,\dots,L-1$. For example,
\[
  \bar{G}_i^* = \frac{1}{N} \sum\limits_{l = 1}^N \hat\Pi_i(\xi_l), \ \ \ i = 0,1,\dots,L-1,
\]
where $\xi_l$ is the $l$th realization and $N$ is the sample size (number of realizations). The value $h$ depending on $L$ specifies the histogram step.

The error of the histogram $\bar{g}^{\langle n \rangle*}$ relative to the density $g$ includes deterministic and stochastic components, and it is estimated from below by the value $\varepsilon_g^{\langle L \rangle*}$:
\[
  \bar\varepsilon_g^{\langle L \rangle*} = \| g_{\nu_1,\nu_2} - \bar{g}_{\nu_1,\nu_2}^{\langle L \rangle*} \|_{L_2(\Omega)} \geqslant \varepsilon_g^{\langle L \rangle*}.
\]

The results of calculations using the formula~\eqref{eqTestEstimatesHist} for both densities from Examples~\ref{ex1} and~\ref{ex2} are given in Tables~\ref{tabErrorG1x} and~\ref{tabErrorG2x}, respectively. They should be compared with the rows ``$*$'' in Tables~\ref{tabErrorG1} and~\ref{tabErrorG2}. Such a comparison shows the undoubted advantage of projection density estimates using Legendre polynomials. The approximation accuracy when using block pulse functions for $L = 64$ corresponds to the approximation accuracy when using Legendre polynomials for $n = 8$ in Example~\ref{ex1} and for $n = 4$ in Example~\ref{ex2}. If the number of block pulse functions $L$ is doubled, then the deterministic component of error decreases by approximately two times, and this conclusion does not depend on the degree of smoothness of the estimated density.

\begin{table}[ht]
\begin{center}
\renewcommand{\arraystretch}{1.1}
\caption{The approximation error $\varepsilon_g^{\langle L \rangle*}$ (Example~\ref{ex1})}\label{tabErrorG1x}
\begin{tabular}{|c|c|c|c|c|}
  \hline
  \ \ \ \ $n = 4$\ \ \ \ &\ \ \ \ $n = 8$\ \ \ \ &\ \ \ $n = 16$\ \ \ &\ \ \ $n = 32$\ \ \ &\ \ \ $n = 64$\ \ \ \\
  \hline
  \hline
  0.186263 & 0.094153 & 0.047203 & 0.023618 & 0.011811 \\
  \hline
\end{tabular}
\end{center}
\end{table}

\begin{table}[ht]
\begin{center}
\renewcommand{\arraystretch}{1.1}
\caption{The approximation error $\varepsilon_g^{\langle L \rangle*}$ (Example~\ref{ex2})}\label{tabErrorG2x}
\begin{tabular}{|c|c|c|c|c|}
  \hline
  \ \ \ \ $n = 4$\ \ \ \ &\ \ \ \ $n = 8$\ \ \ \ &\ \ \ $n = 16$\ \ \ &\ \ \ $n = 32$\ \ \ &\ \ \ $n = 64$\ \ \ \\
  \hline
  \hline
  0.171305 & 0.089038 & 0.044945 & 0.022526 & 0.011270 \\
  \hline
\end{tabular}
\end{center}
\end{table}

The problem of the conditional optimization of the algorithm for obtaining the histogram was solved in~\cite{Che_DAN62}. In this problem, optimal parameters satisfy the relationship $N_\mathrm{opt} \asymp L_\mathrm{opt}^3$ which corresponds to the inequality~\eqref{eqEstimateL2C}. A generalization of this result in the context of stochastic differential equations can be found in~\cite{Ave_Math24}. Choosing parameters in this way, the computational error $\bar\varepsilon_g^{\langle L \rangle*}$ will be approximately twice as large as its deterministic component $\varepsilon_g^{\langle L \rangle*}$, since the conditional optimization of the algorithm for obtaining the histogram assumes the equalization of deterministic and stochastic components.

For densities from Examples~\ref{ex1} and~\ref{ex2}, to reduce the approximation error using the histogram by two times, it is necessary to increase the number of block pulse functions $L$ by $2$ times and the sample size $N$ by $2^3 = 8$ times. The projection estimates of densities using Legendre polynomials are more effective in these examples. To reduce the approximation error by two times in Example~\ref{ex1}, it suffices to increase the projection expansion length $n$ by $2^{2/3} \approx 1.587$ times and the sample size $N$ by $(2^{2/3})^4 = 2^{8/3} \approx 6.350$ times. In Example~\ref{ex2}, it suffices to increase the projection expansion length $n$ by $2^{2/5} \approx 1.320$ times and the sample size $N$ by $(2^{2/5})^6 = 2^{12/5} \approx 5.278$ times. In this case, increasing $n$ and $N$ implies their subsequent rounding-up.

\section{Conclusions}\label{secConcl}

A distinctive feature of algorithms proposed in this paper is the jointly obtaining projection estimates of the density and distribution function of a random variable. In addition, the problem of the conditional optimization of algorithms is solved under the assumption that estimated functions belong to the Sobolev--Slobodetskij space. This assumption allows one to use a more precise estimate of the convergence rate for projection expansions compared with the case where the ordinary Sobolev space or the space of continuously differentiable functions are used. Numerical experiments confirm the obtained theoretical results.

The proposed algorithms and conditionally optimal relationships between parameters ensure the minimum complexity when obtaining projection estimates to achieve the required approximation accuracy. These algorithms are important for studying the properties of solutions to stochastic differential equations~\cite{AveRyb_NAA24, Gio_Aero25}, including stochastic differential equations with a Poisson component~\cite{AveRyb_RJNAM20, SongWu_JOTA23} or with the switching right-hand side~\cite{AveRyb_RJNAM07, Bor_Math24}.

\end{document}